\documentclass [journal]{IEEEtran}

\usepackage{amsmath,amssymb}
\usepackage{amsthm}
\usepackage[bookmarks,colorlinks]{hyperref} 
\hypersetup{colorlinks,citecolor= red,filecolor= blue,linkcolor= blue,urlcolor=blue}
\usepackage{acronym} 
\usepackage{graphicx,psfrag,cite,subcaption}
\usepackage{color} 
\usepackage{cite}
\usepackage[normalem]{ulem}
\usepackage[usenames,dvipsnames,table]{xcolor}
\usepackage[export]{adjustbox}
\usepackage{algorithm}
\usepackage{algpseudocode}
\usepackage{subcaption}
\usepackage{winsnotation}

\usepackage{cancel}
\usepackage{bigints}
\usepackage{mathtools}
\newtheorem{theorem}{Theorem}
\newtheorem{lem}{Lemma}
\newtheorem{remark}{Remark}
\DeclareMathOperator{\sinc}{sinc}

\newcommand{\underbracee}[2]{\underbrace{#2}_{#1}}

\acrodef{PI}[PI]{Principal Investigator}
\acrodef{KI}[KI]{Key Investigator}
\acrodef{RO}[RO]{research objective}
\acrodef{RA}[RA]{research activity}
\acrodefplural{RA}[RA]{research activities}

\acrodef{AFOSR}[AFOSR]{Air Force Office of Scientific Research}
\acrodef{ARO}[ARO]{Army Research Office}
\acrodef{DARPA}[DARPA]{Defense Advanced Research Projects Agency}
\acrodef{DOD}[DoD]{Department of Defense}
\acrodef{ONR}[ONR]{Office of Naval Research}
\acrodef{ISR}[ISR]{intelligence, surveillance and reconnaissance}
\acrodef{SA}[SA]{situational awareness}
\acrodef{MDMP}[MDMP]{military decision making process}
\acrodef{COA}[COA]{course of action}
\acrodef{OODA}[OODA]{observe, orient, decide, and act}

\acrodef{AA}[AeroAstro]{Aeronautics and Astronautics}
\acrodef{COC}[COC]{College of Computing}
\acrodef{MIT}[MIT]{Massachusetts Institute of Technology}
\acrodef{IDSS}[IDSS]{Institute for Data, Systems, and Society}
\acrodef{ISN}[ISN]{Institute for Soldier Nanotechnologies}
\acrodef{LIDS}[LIDS]{Laboratory for Information and Decisions Systems}
\acrodef{WINSLAB}[WINSLab]{Wireless Information and Network Sciences Laboratory}

\acrodef{IIT}[IIT]{Indian Institute of Technology}
\acrodef{UA}[UA]{University of Arizona}

\acrodef{LIDAR}[LIDAR]{light detection and ranging}
\acrodef{NSTC}[NSTC]{National Science \& Technology Council}
\acrodef{NLA}[NLA]{noiseless linear amplification}
\acrodef{QIS}[QIS]{quantum information science}
\acrodef{QN}[QN]{quantum network}
\acrodef{QNS}[QNS]{quantum network science}
\acrodef{QSD}[QSD]{Quantum State Discrimination}
\acrodef{RADAR}[RADAR]{radio detection and ranging}

\acrodef{ACK}[ACK]{acknowledge}
\acrodef{AE}[AE]{angle estimate}
\acrodef{AI}[AI]{angle information}
\acrodef{AII}[AII]{angle information intensity}
\acrodef{AIS}[AIS]{automatic identification system}
\acrodef{AL}[AL]{angle likelihood}
\acrodef{ALU}[ALU]{arithmetic logic unit}
\acrodef{ANN}[ANN]{artificial neural network}
\acrodef{AOA}[AoA]{angle-of-arrival}
\acrodef{API}[API]{application programming interface}
\acrodef{ASIC}[ASIC]{application-specific integrated circuit}
\acrodef{AWGN}[AWGN]{additive white Gaussian noise}
\acrodef{AWS}[AWS]{Amazon Web Services}
\acrodef{BC}[BC]{belief condensation}
\acrodef{BCF}[BCF]{belief condensation filter}
\acrodef{BLE}[BLE]{Bluetooth Low Energy}
\acrodef{BP}[BP]{belief propagation}
\acrodef{BPZF}[BPZF]{band-pass zonal filter}
\acrodef{BTB}[BTB]{Bellini-Tartara bound}
\acrodef{CCDF}[CCDF]{complementary cumulative distribution function}
\acrodef{CDF}[CDF]{cumulative distribution function}
\acrodef{CE}[CE]{cross-entropy}
\acrodef{CEO}[CEO]{counting error outage}
\acrodef{CESS}[CE-SS]{cross-entropy SS}
\acrodef{CF}[CF]{characteristic function}
\acrodef{CFR}[CFR]{channel frequency response}
\acrodef{CIR}[CIR]{channel impulse response}
\acrodef{CNS}[CNS]{classical network science}
\acrodef{CR}[CR]{channel response}
\acrodef{CRB}[CRB]{Cram\'{e}r-Rao bound}
\acrodef{CRLB}[CRLB]{Cram\'{e}r-Rao lower bound}
\acrodef{CSI}[CSI]{channel state information}
\acrodef{CV}[CV]{continuous variable}
\acrodef{DA}[DA]{data-association}
\acrodef{DE}[DE]{distance estimate}
\acrodef{DPOMDP}[Dec-POMDP]{cecentralized-partially observed Markov decision process}
\acrodef{DFE}[DFE]{decision feedback equalizer}
\acrodef{DP}[DP]{direct path}
\acrodef{DRT}[DRT]{distance ratio test}
\acrodef{DV}[DV]{discrete variable}
\acrodef{EBIT}[EBIT]{entanglement bit}
\acrodef{ED}[ED]{energy detector}
\acrodef{EED}[EED]{Engineering Every Day}
\acrodef{EKF}[EKF]{extended Kalman filter}
\acrodef{EIRP}[EIRP]{equivalent isotropically radiated power}
\acrodef{ESD}[ESD]{energy-based soft-decision}
\acrodef{ESPRIT}[ESPRIT]{estimation of signal parameters via rotational invariant techniques}
\acrodef{FCC}[FCC]{Federal Communications Commission}
\acrodef{FG}[FG]{factor graph}
\acrodef{FII}[FII]{Fisher information inequality}
\acrodef{FIM}[FIM]{Fisher information matrix}
\acrodef{FL}[FL]{feature likelihood}
\acrodef{FP}[FP]{feature potential}
\acrodef{FSC}[FSC]{finite state controller}
\acrodef{FW}[FW]{Fisher-Wald}
\acrodef{FY}[FY]{fiscal year}
\acrodef{GDOP}[GDOP]{geometric dilution of precision}
\acrodef{GLMB}[GLMB]{generalized labeled multi-Bernoulli}
\acrodef{GLRT}[GLRT]{generalized likelihood ratio test}
\acrodef{GMMLMB}[G-MM-LMB]{Gaussian MM-LMB}
\acrodef{GNSS}[GNSS]{global navigation satellite system}
\acrodef{GPS}[GPS]{Global Positioning System}
\acrodef{HCA}[HCA]{heterogeneous computing architecture}
\acrodef{HDSA}[HDSA]{high-definition situation-aware}
\acrodef{HDP}[HDP]{hierarchical Dirichlet process}
\acrodef{HI}[HI]{hard information}
\acrodef{HMM}[HMM]{hidden Markov model}
\acrodef{IID}[IID]{independent, identically distributed}
\acrodef{IMU}[IMU]{inertial measurement unit}
\acrodef{INDFT}[INDFT]{inverse non-uniform discrete Fourier transform}
\acrodef{INR}[INR]{interference-to-noise ratio}
\acrodef{IOT}[IoT]{Internet-of-Things}
\acrodef{IR-UWB}[IR-UWB]{impulse radio UWB}
\acrodef{IRPR}[iRPR]{infinite regionalized policy representation}
\acrodef{JBSF}[JBSF]{jump back and search forward}
\acrodef{KDE}[KDE]{kernel density estimation}
\acrodef{KF}[KF]{Kalman filter}
\acrodef{KL}[KL]{Kullback-Leibler}
\acrodef{LBP}[LBP]{loopy belief propagation}
\acrodef{LEM}[LEM]{Laplacian eigen-map}
\acrodef{LEO}[LEO]{localization error outage}
\acrodef{LMB}[LMB]{labeled multi-Bernoulli}
\acrodef{LMS}[LMS]{least means square}
\acrodef{LOCC}[LOCC]{local operations and classical communication}
\acrodef{LOS}[LOS]{line-of-sight}
\acrodef{LOT}[LoT]{Localization-of-Things}
\acrodef{LRT}[LRT]{likelihood ratio test}
\acrodef{LRFS}[LRFS]{Labeled random finite set}
\acrodef{LS}[LS]{least squares}
\acrodef{LSE}[LSE]{line spectral estimation}
\acrodef{MAC}[MAC]{medium access control}
\acrodef{MAP}[MAP]{maximum a posteriori}
\acrodef{MBS}[MBS]{maximum bin search}
\acrodef{MC}[MC]{Monte Carlo}
\acrodef{MDD}[MDD]{minimum distance distribution}
\acrodef{MDP}[MDP]{Markov decision process}
\acrodef{MF}[MF]{matched filter}
\acrodef{MHT}[MHT]{multi-hypothesis tracking}
\acrodef{MIMO}[MIMO]{multiple-input multiple-output}
\acrodef{ML}[ML]{maximum likelihood}
\acrodef{MLE}[MLE]{maximum likelihood estimation}
\acrodef{MMSE}[MMSE]{minimum-mean-square-error}
\acrodef{MMLMB}[MM-LMB]{merged-measurement LMB}
\acrodef{MOT}[MOT]{multi-object tracking}
\acrodef{MOU}[MOU]{measurement origin uncertainty}
\acrodef{MP}[MP]{map potential}
\acrodef{MSE}[MSE]{mean-square error}
\acrodef{MUI}[MUI]{multi-user interference}
\acrodef{MUSIC}[MUSIC]{multiple signal classification}
\acrodef{NBI}[NBI]{narrowband interference}
\acrodef{NISQ}[NISQ]{Noisy Intermediate-Scale Quantum}
\acrodef{NLN}[NLN]{network localization and navigation}
\acrodef{NLOS}[NLOS]{non-line-of-sight}
\acrodef{NSF}[NSF]{National Science Foundation}
\acrodef{OOT}[OoT]{ocean-of-things}
\acrodef{OOTCESS}[OoT-CESS]{OoT-CESS}
\acrodef{OP}[OP]{outage probability}
\acrodef{OSPA}[OSPA]{optimum subpattern assignment}
\acrodef{P-Max}[P-Max]{$P$-Max}  
\acrodef{PAR}[PAR]{probabilistic association rule}
\acrodef{PCA}[PCA]{principal component analysis}
\acrodef{PDF}[PDF]{probability distribution function}
\acrodef{PDP}[PDP]{power delay profile}
\acrodef{PEB}[PEB]{position error bound}
\acrodef{PF}[PF]{physical features}
\acrodef{PHD}[PHD]{probability hypothesis density}
\acrodef{PMF}[PMF]{probability mass function}
\acrodef{POCS}[POCS]{projection onto convex sets}
\acrodef{POMDP}[POMDP]{partially observed Markov decision process}
\acrodef{PPM}[PPM]{pulse position modulation}
\acrodef{PPP}[PPP]{Poisson point process}
\acrodef{QUA}[QuaDRiGa]{QUAsi Deterministic RadIo channel GenerAtor}
\acrodef{RCS}[RCS]{radar cross section}
\acrodef{RF}[RF]{radiofrequency}
\acrodef{RFID}[RFID]{radio frequency identification}
\acrodef{RFS}[RFS]{random finite set}
\acrodef{RI}[RI]{range information}
\acrodef{RII}[RII]{range information intensity}
\acrodef{RL}[RL]{range likelihood}
\acrodef{RLS}[R-LS]{range-based least squares}
\acrodef{RLMC}[RL-MC]{reinforcement learning Monte Carlo}
\acrodef{RM}[RM]{resource management}
\acrodef{RMS}[RMS]{root mean square}
\acrodef{RMSE}[RMSE]{root-mean-square error}
\acrodef{ROI}[ROI]{region of interest}
\acrodef{RPR}[RPR]{regionalized policy representation}
\acrodef{RRC}[RRC]{root raised cosine}
\acrodef{RSS}[RSS]{received signal strength}
\acrodef{RTT}[RTT]{round-trip time}
\acrodef{RV}[RV]{random variable}
\acrodef{SBS}[SBS]{serial backward search}
\acrodef{SBSMC}[SBSMC]{serial backward search for multiple clusters}
\acrodef{SC}[SC]{soft constraint}
\acrodef{SDN}[SDN]{software defined network}
\acrodef{SK}[SK]{soft knowledge}
\acrodef{SI}[SI]{soft information}
\acrodef{SII}[SII]{speed information intensity}
\acrodef{SIR}[SIR]{signal-to-interference ratio}
\acrodef{SLAM}[SLAM]{simultaneous localization and mapping}
\acrodef{SMC}[SMC]{sequential monte carlo}
\acrodef{SNR}[SNR]{signal-to-noise ratio}
\acrodef{SO}[SO]{soft observation}
\acrodef{SPAWN}[SPAWN]{sum-product algorithm over a wireless network}
\acrodef{SPEB}[SPEB]{squared position error bound}
\acrodef{SR}[SR]{sensor radar}
\acrodef{SS}[SS]{sensor selection}
\acrodef{SSCH}[SSh]{sensor scheduling} 
\acrodef{STEM}[STEM]{science, technology, engineering, and mathematics}
\acrodef{SVE}[SVE]{single-value estimate}
\acrodef{TBD}[TBD]{track before detect}
\acrodef{TCS}[TCS]{threshold crossing search}
\acrodef{TD}[TD]{temporal difference}
\acrodef{TDOA}[TDOA]{time difference-of-arrival}
\acrodef{TH}[TH]{time-hopping}
\acrodef{TNR}[TNR]{threshold-to-noise ratio}
\acrodef{TOA}[TOA]{time-of-arrival}
\acrodef{TOF}[TOF]{time-of-flight}
\acrodef{TSD}[TSD]{threshold-based soft-decision}
\acrodef{TWS}[TWS]{track while scan}
\acrodef{UAV}[UAV]{unmanned aerial vehicles}
\acrodef{UKF}[UKF]{unscendent Kalman filter}
\acrodef{UML}[UML]{unsupervised machine learning}
\acrodef{UWB}[UWB]{ultra-wideband}
\acrodef{VM}[VM]{von Mises}
\acrodef{VNA}[VNA]{vector network analyzer}
\acrodef{WAF}[WAF]{wall attenuation factor}
\acrodef{WBI}[WBI]{wideband interference}
\acrodef{WED}[WED]{wall extra delay}
\acrodef{WLS}[WLS]{weighted least squares}
\acrodef{WPAN}[WPAN]{wireless personal area network}
\acrodef{WSN}[WSN]{wireless sensor network}
\acrodef{WWB}[WWB]{Weiss-Weinstein bound}
\acrodef{XGL}[xGL]{next generation localization}
\acrodef{ZZB}[ZZB]{Ziv-Zakai bound}
\acrodef{ZZLB}[ZZLB]{Ziv-Zakai lower bound}

\acrodef{SM}[SM]{spatial modulation}
\acrodef{MD}[MD]{mobile device}
\acrodef{BS}[BS]{base station}
\acrodef{LoS}[LoS]{Line-of-Sight}
\acrodef{NLoS}[NLoS]{Non-Line-of-Sight}
\acrodef{mmWave}[mmWave]{millimeter-wave}
\acrodef{NOMA}[NOMA]{Non-Orthogonal Multiple Access}
\acrodef{5G}[5G]{fifth generation}
\acrodef{ToA}[ToA]{time-of-arrival}
\acrodef{TDoA}[TDoA]{time-difference-of-arrival}
\acrodef{DoA}[DoA]{direction-of-arrival}
\acrodef{OFDM}[OFDM]{Orthogonal-Frequency-Division-Multiplexing}
\acrodef{DAS}[DAS]{distributed antenna system}
\acrodef{RRU}[RRU]{remote radio unit}
\acrodef{2D}[2D]{two-dimensional}
\acrodef{ML}[ML]{maximum-likelihood}
\acrodef{UAV}[UAV]{Unmanned-Aerial-Vehicle}

\usepackage[yyyymmdd,hhmmss]{datetime} 
\newdateformat{monthyeardate}{\monthname[\THEMONTH] \THEDAY, \THEYEAR} 

\newcommand{\circledone}{\raisebox{.5pt}{\textcircled{\raisebox{-.9pt} {1}}}}
\newcommand{\circledtwo}{\raisebox{.5pt}{\textcircled{\raisebox{-.9pt} {2}}}}
\newcommand{\circledthree}{\raisebox{.5pt}{\textcircled{\raisebox{-.9pt} {3}}}}
\newcommand{\circledfour}{\raisebox{.5pt}{\textcircled{\raisebox{-.9pt} {4}}}}

\newcommand{\paperTitle}{Direction-Shift Keying for High-Mobility mmWave Communications: Geometry-Based Resilience to Short Channel Coherence Times and Phase Noise}

\begin{document}


\title{\paperTitle}


\author{	
  Mohaned Chraiti, Ozgur Ercetin, Ali Ghrayeb, and Ali Gorcin

\thanks{
M. Chraiti and O. Ercetin are with the Electronics Engineering Department, Sabanci University, Istanbul, Turkey (e-mails: \texttt{mohaned.chraiti@sabanciuniver.edu} and \texttt{oercetin@sabanciuniv.edu}.)\\ 
A. Ghrayeb is with College of Science and Engineering, Hamad Bin Khalifa University, Doha, Qatar (
Email: \texttt{aghrayeb@hbku.edu.qa}). \\
A. Gorcin is with Communications and Signal Processing Research (HİSAR) Lab., T{\"{U}}B{\.{I}}TAK B{\.{I}}LGEM, Kocaeli, Turkey and Department of Electronics and Communication Engineering, Istanbul Technical University, {\.{I}}stanbul, Turkey (Email: \texttt{aligorcin@itu.edu.tr})
}
	
}

\maketitle


\markboth{\qquad\qquad\qquad\qquad\qquad\qquad\qquad\qquad\qquad\qquad\qquad\qquad\qquad\qquad\qquad\qquad\qquad\qquad\qquad\qquad\qquad\qquad\qquad\qquad}
		{ \paperTitleMarkboth}



\setcounter{page}{1}


\begin{abstract}
Communication reliability in high-mobility millimeter-wave (mmWave) systems is fundamentally challenged by rapid channel variations, manifested in short channel coherence times, and phase noise, which necessitate frequent channel estimation and phase tracking. To overcome these limitations, we introduce Direction-Shift Keying (DSK), a geometry-based modulation scheme that exploits the relative stability of propagation directions rather than the rapidly varying amplitude and phase of the wireless channel. DSK conveys information through distinguishable directions of arrival (DoAs), enabled by a distributed antenna system (DAS), thereby shifting the communication principle from channel-dependent to geometry-dependent signaling. We first derive the structure of the optimal DSK detector for a mobile device equipped with an arbitrary number of antennas. We then establish the fundamental law governing the temporal stability of propagation directions, termed Direction Coherence Time (DCT), and show that DCT scales as $d/v$, where $d$ is the transmitter--receiver distance and $v$ is the relative velocity. In contrast, conventional Channel Coherence Time (CCT) scales as $\lambda/v$, where $\lambda$ denotes the wavelength, revealing a coherence-time advantage on the order of $d/\lambda$, potentially exceeding four orders of magnitude in mmWave systems. Furthermore, we establish the inherent invariance of DSK to common phase rotations, providing resilience to common oscillator phase noise without explicit phase tracking. Numerical results corroborate the theoretical findings and demonstrate the potential of geometry-based signaling for reliable high-mobility mmWave communications.
\end{abstract}

\begin{IEEEkeywords}
 Channel coherence time, distributed antennas system, localization,  location-shift keying, mmWave communications, phase noise, space-shift keying.
\end{IEEEkeywords}

\section{Introduction}\label{Sec:Introduction}

\subsection{Motivation}

Reliable communication at millimeter-wave (mmWave) frequencies becomes particularly challenging under high mobility. As a mobile device moves, the wireless channel undergoes rapid variations, resulting in short channel coherence times and requiring frequent channel estimation to sustain coherent detection. This challenge is further aggravated by oscillator phase noise, whose impact becomes increasingly pronounced at high carrier frequencies \cite{HexaX,HexaX1,PhaseNoise3GPP,DuVal,ChiVenVal:3,ChiVenVal:1,VaChoiHeath}. The combined effects impose substantial signaling overhead and stringent requirements on receiver synchronization and channel tracking, potentially limiting the benefits of the large bandwidth available at mmWave frequencies.

The Channel Coherence Time (CCT) is fundamentally governed by the carrier wavelength $\lambda$ and the mobile device velocity $v$, with a characteristic scaling of $\lambda/v$. For instance, consider a vehicle traveling at $100\,\mathrm{km/h}$ at a carrier frequency of $30\,\mathrm{GHz}$. Using the conventional coherence-time approximation, we obtain
\begin{equation}
t_{\mathrm{CCT}} \simeq \frac{9}{16\pi}\frac{\lambda}{v}
\approx 64\,\mu\mathrm{s},
\end{equation}
as characterized in \cite[Sec. II-A, (2)]{TelTse:00}. Such a short interval illustrates the difficulty of maintaining accurate channel information under high mobility, particularly when channel-dependent detection requires frequent reference-signal updates.

Nevertheless, the rapid evolution of the wireless channel does not necessarily imply an equally rapid evolution of its underlying propagation geometry. While small displacements relative to the wavelength can substantially alter channel phases, the corresponding propagation directions may remain nearly unchanged. In particular, the temporal variation of the Direction-of-Arrival (DoA) is governed by the relative geometry of the transmitter and receiver, whose characteristic spatial scale can be orders of magnitude larger than the wavelength. This separation between channel coherence and geometric stability opens an alternative approach to wireless communication, in which information is conveyed through relatively persistent spatial signatures rather than rapidly varying channel coefficients.

Motivated by this principle, we present Direction-Shift Keying (DSK), a geometry-based modulation scheme designed for resilient high-mobility mmWave communications. DSK conveys information through distinguishable propagation directions generated by spatially distributed transmit antennas. A multi-antenna mobile receiver identifies the active transmitter from its geometric signature, represented by the Time-Difference-of-Arrival (TDoA) across receive antennas. Unlike conventional channel-dependent detection, DSK exploits the relative stability of these signatures and avoids the need to continuously estimate the instantaneous complex channel coefficients. Furthermore, its reliance on relative spatial information provides inherent invariance to common phase rotations, offering resilience to oscillator phase noise.

The fundamental advantage of DSK therefore lies in replacing the rapidly varying channel coefficients with geometric signatures whose temporal validity can be substantially longer. This raises three central questions: how to recover the transmitted information from geometric signatures, how long these signatures remain reliable under mobility, and to what extent geometry-based detection is resilient to oscillator phase noise. Addressing these questions establishes the theoretical foundations of DSK for high-mobility mmWave communications.

\subsection{Related Work}

The limitations imposed by short channel coherence times have motivated extensive research on channel estimation and training overhead in mmWave systems \cite{DuVal,ChiVenVal:3,ChiVenVal:1,VaChoiHeath}. In \cite{DuVal}, the authors demonstrated that channel training can consume a substantial fraction of the available communication resources, potentially limiting the exploitable bandwidth despite the multi-gigahertz spectrum available at mmWave frequencies. The trade-off between channel estimation accuracy and training overhead was further investigated in \cite{ChiVenVal:3,ChiVenVal:1}, highlighting the increasing cost of channel acquisition under challenging propagation conditions.

A related line of research exploits the slower evolution of spatial directions for beam management. In \cite{VaChoiHeath,ZhaZhaZha,MohanedOzgur}, the authors showed that beamforming coherence times can substantially exceed conventional channel coherence times. Although this reduces the frequency of beam alignment, conventional coherent detection may still require estimation of the effective channel coefficient after beamforming. DSK takes a different approach by employing the propagation direction itself as the information-bearing signature, rather than using it solely to maintain beam alignment.

Oscillator phase noise constitutes another major impairment in mmWave and higher-frequency communications. Nonideal oscillators introduce random phase variations that distort received signals and degrade conventional coherent detection \cite{HexaX,HexaX1,PhaseNoise3GPP,PhaseNoise:ChungPra}. Existing solutions typically rely on improved oscillator hardware, phase tracking, and additional reference signals, introducing implementation complexity and signaling overhead \cite{PhaseNoise:Song,MohanedDidem,MohanedNFreqMulti}. These limitations motivate signaling techniques that are intrinsically less dependent on absolute received phase.

Spatial Modulation (SM) offers an alternative means of conveying information through antenna indices while reducing transmitter hardware complexity \cite{RenHaaGhr,JegGhrSzc:J1}. In particular, Space-Shift Keying (SSK) maps information bits to the index of the active transmit antenna. Conventional coherent SSK detection identifies this index through the corresponding channel coefficients, requiring sufficiently accurate Channel State Information (CSI) at the receiver \cite{JegGhrSzc:J2}. Consequently, its performance can deteriorate when channel estimates become outdated or when uncompensated phase noise affects detection.

\subsection{Contributions}
The proposed Direction-Shift-Keying (DSK) preserves the antenna-index signaling principle but replaces channel-coefficient-based detection with geometry-based identification. This distinction is particularly significant under high mobility, where propagation directions may remain distinguishable long after instantaneous channel estimates lose their validity.

This paper establishes a geometry-based communication framework for high-mobility mmWave systems and its analytical characterization. The central principle is that the temporal stability of propagation geometry can be exploited to sustain information detection beyond the conventional channel coherence time. We characterize this advantage through the Direction Coherence Time (DCT), which measures the duration over which geometric signatures remain sufficiently stable for reliable DSK detection. The main contributions are summarized as follows:

\begin{itemize}
\item We introduce the DSK transmission and detection framework to mobile devices equipped with an arbitrary number of receive antennas and derive a multi-antenna correlation-based detection rule from the received-signal model.

\item We establish an analytical characterization of DCT under receiver mobility, relating the evolution of geometric signatures to transmitter--receiver distance, mobile velocity, receiver antenna geometry, and signal bandwidth. In contrast to the wavelength-dependent CCT, the resulting DCT characterization reveals the potential for orders-of-magnitude longer reference validity under mmWave operating conditions.

\item We characterize the invariance of DSK detection to common phase rotations, establishing its resilience to common oscillator phase impairments without explicit common-phase tracking.

\item We evaluate the performance of DSK against conventional SSK through numerical simulations incorporating high mobility, phase noise, and pilot-aided reference estimation. The results demonstrate the advantages of geometric signatures in maintaining detection reliability while reducing the frequency of reference-signal updates.
\end{itemize}

The rest of this paper is structured as follows. In Sec. \ref{Sec:SysDes}, we provide the system and transmission models and derive the optimal detector. In Sec. \ref{Sec:PostionCoherenceTime}, we provide the theoretical analysis of the proposed approach. Sec. \ref{Sec:SimulationResults} and Sec. \ref{Sec:Conclusions} are dedicated to simulation results and concluding remarks, respectively.

\section{Direction Shift-Keying}\label{Sec:SysDes}
\subsection{System Model}

We consider a high-mobility mmWave communication system in which information transmission exploits the relative stability of propagation geometry rather than instantaneous channel coefficients. The underlying architecture consists of a DAS communicating with a moving multi-antenna device, as illustrated in Fig.~\ref{Fig:SysMod}. Spatially distributed transmit antennas establish distinct propagation directions toward the mobile device, forming a set of geometric signatures that can be used to convey information. Although receiver motion induces rapid variations in the amplitude and phase of the received signals, the relative propagation directions may evolve on a substantially longer timescale. DSK exploits this separation by mapping information bits to transmit antenna indices and recovering them through the corresponding time difference of arrival (TDoA) signatures at the receiver.

We consider a mmWave DAS in which the base station (BS) antennas are spatially sparse.\footnote{The DAS is increasingly considered for mmWave deployments \cite{JDAS1,JDAS2,ORAN_nGRG_2024_mWAD}. Using spatially distributed (i.e., sparse) remote antennas has been shown to improve spectral efficiency, extend coverage, and increase diversity gain. }  A representative configuration is depicted in Fig.~\ref{Fig:SysMod}. The BS is equipped with \( M \) distributed transmit antennas, denoted by \( \{\mathrm{T}_1, \mathrm{T}_2, \dotsc, \mathrm{T}_M\} \), each located at a position with coordinates \( \mathbf{b}_m = [b_{\mathrm{x},m}, b_{\mathrm{y},m}] \), where \( m \in \{1, 2, \dotsc, M\} \). The MD is equipped with \( N \geq 2 \) antennas.
We consider the case of a MD in motion. At a given instant, the position of the device is defined by the coordinates of its geometric center, denoted as $\bar{\mathbf{a}} = [\bar{a}_{\mathrm{x}}, \bar{a}_{\mathrm{y}}]$. The position of its $n$-th antenna element, relative to this center, is given by
$$
\mathbf{a}_n = l_n \begin{bmatrix} \cos(\phi_n) \\ \sin(\phi_n) \end{bmatrix},
$$
where $l_n$ represents the radial distance from the center, and $\phi_n$ is the orientation angle of the antenna.

The transmission process follows the SSK principle: only one BS antenna, indexed by \( v \in \{1, \dotsc, M\} \), is active per symbol interval and transmits a known unmodulated pulse \( s(t) \) with power \( P \) over a duration \( T \). The transmit signal from the \( m \)-th BS antenna is defined as
\begin{equation}
x_m(t)=\begin{cases}
s(t), & \text{if } m = v, \\
0, & \text{otherwise}.
\end{cases}
\end{equation}
The signal emitted by the active antenna \( \mathrm{T}_v \) reaches the \( n \)-th MD antenna with a delay \( \tau_n^v \), corresponding to the strongest propagation path.\footnote{In line with \cite{Mohaned}, this work focuses on the dominant ray due to the sparse nature of mmWave channels. Extensions to multi-ray scenarios can potentially further improve DSK performance, as discussed in \cite{QiKobSud:VT}.} The transmisions are subject to a random carrier frequency offset \( \Delta f \) due to oscillator mismatch, introducing an unknown phase rotation. The baseband received signal at the \( n \)-th MD antenna is expressed as \cite{Mohaned}
\begin{equation}\label{eq:01}
\rv{r}_n(t) = \rho \, \mathrm{e}^{-\jmath 2\pi f_c \tau_n^v} \, \mathrm{e}^{-\jmath 2\pi \Delta f t} \, s(t - \tau_n^v) + \rv{w}_n(t),
\end{equation}
where \( \rho \) is the path gain and \( \rv{w}_n(t) \sim \mathcal{CN}(0, \sigma^2) \) denotes complex additive white Gaussian noise.

\begin{figure}[t]
\centering
\includegraphics[scale=0.19]{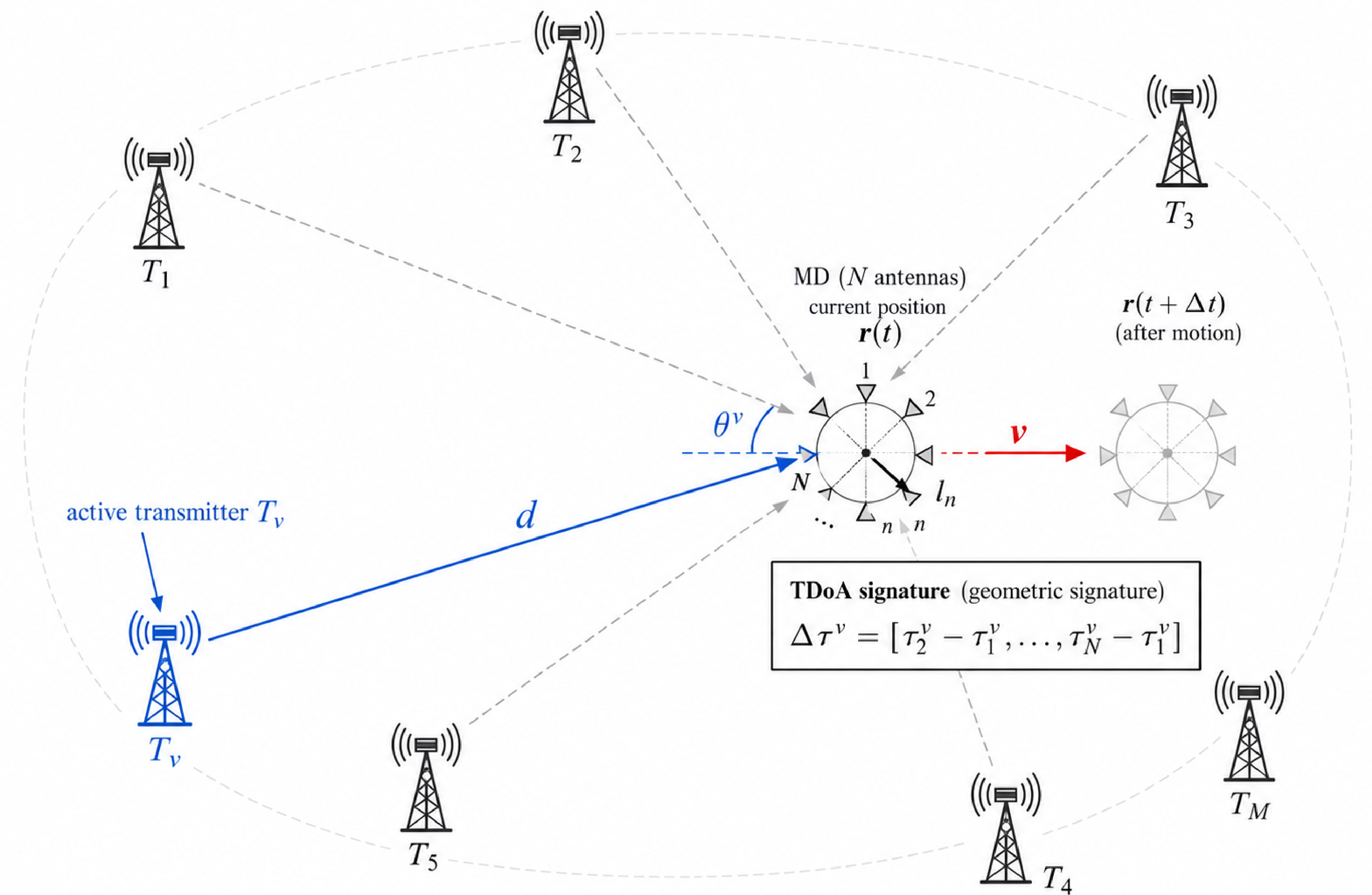}
\caption{Illustration of the mmWave DAS model.}
\label{Fig:SysMod}
\vspace{-0.5cm}
\end{figure}
We assume that neither the BS nor the MD has prior knowledge of the instantaneous channel coefficients $\{\alpha_n\}_{n=1}^N$ or the oscillator-induced phase noise $\Delta f$. The MD is, however, assumed to have access to a Time-Difference-of-Arrival (TDoA) profile associated with each BS antenna. This TDoA profile serves as a spatial fingerprint that enables inference of the DoA during the detection process. The TDoA generating vector for BS-antenna \( T_m \) is denoted as
\[
\boldsymbol{\Delta\tau}^m = [\tau^m_2 - \tau^m_1, \dotsc, \tau^m_N - \tau^m_1].
\]
This vector fully specifies the TDoA structure, from which the delay between any antenna pair \( (l,k) \) can be reconstructed as:
\[
\tau^m_l - \tau^m_k = \boldsymbol{\Delta\tau}^m(l) - \boldsymbol{\Delta\tau}^m(k).
\]
Without loss of generality, we assume \( M \) is a power of two to facilitate binary bit-to-antenna mapping. While the analysis focuses on the downlink for clarity, the proposed DSK scheme and its performance metrics can be readily adapted to uplink configurations through appropriate beam-steering techniques.

\subsection{Transmission and Detection Processes}

At each symbol period, one BS antenna transmits an unmodulated pulse \( s(t) \), identical across antennas and time slots in waveform, amplitude, and phase, akin to pilot signals. Bit sequences are segmented into blocks of \( \log_2(M) \) bits and mapped to transmit antenna indices following the SSK paradigm \cite{JegGhrSzc:J1}. Unlike SSK, DSK employs a direction-based detection strategy. The receiver can infers the active antenna index by: (i) estimating the TDoA vector from received signals, and (ii) applying maximum likelihood (ML) detection using pre-stored TDoA signatures. This approach bypasses channel coefficient estimation, operating independently of CSI. The optimal ML detector is derived in the next section.

\subsection{Optimal Detector}\label{Sec:OneStageDetection}

Define
\[
\alpha^{v}_n \overset{\Delta}{=} \rho \mathrm{e}^{-\jmath 2\pi f_c \tau_n^v} \mathrm{e}^{\jmath 2\pi \Delta f t}
\]
as the complex channel-path coefficient incorporating phase noise. The received signal in \eqref{eq:01} becomes
\begin{equation}\label{eq:RecievedSignal}
\rv{r}_n(t) = \alpha^{v}_n s(t - \tau^v_n) + \rv{w}_n(t).
\end{equation}
Let $H_m$ denote the hypothesis that BS antenna $m$ is active:
\[
H_m: \rv{r}_n(t) \sim \mathcal{CN}(\alpha_n s(t - \tau^m_n), \sigma^2), \quad \forall n \in [1, N].
\]
The ML detector selects the antenna index that maximizes the conditional likelihood of the received vector $\RV{r} = [\rv{r}_1, \dotsc, \rv{r}_N]$ given the known TDoA signature $\boldsymbol{\Delta\tau}^m = [\tau^m_2 - \tau^m_1, \dotsc, \tau^m_N - \tau^m_1]$:
\begin{equation}\label{eq:ML}
\tilde{m} = \arg\max_{m \in \{1,\dotsc,M\}} f_{\RV{r}|\boldsymbol{\Delta\tau}}(\RV{r}|\boldsymbol{\Delta\tau}^m).
\end{equation}
Since $\RV{r}$ has dimension $N$ while $\boldsymbol{\Delta\tau}^m$ has dimension $N-1$, standard factorization into independent marginals is not possible. Now, using the chain rule:
\begin{equation}\label{eq:DetectorMLTDoA}
\begin{aligned}
&f_{\RV{r}|\RV{\Delta}\RV{\tau}}(\V{r}|\V{\Delta}\V{\tau}^m)\\&=f_{\RV{r}|\RV{\Delta}\RV{\tau}}({r}_1,\dotsc,{r}_n|\tau^m_2-\tau^m_1,\dotsc, \tau^{m}_{N}-\tau^m_1)\\
&{=}\prod_{n=1}^{N}f_{\rv{r}_n|\RV{\Delta}\RV{\tau}, \rv r_{1\to n-1}}(r_{n}|\V{\Delta}\V{\tau}^{m}_{n,1\to n-1},\V{r}_{1\to n-1}),
\end{aligned}
\end{equation}
where $\RV{r}_{1\to n-1} = [\rv{r}_1, \dotsc, \rv{r}_{n-1}]$ and $\boldsymbol{\Delta\tau}^{m}_{n,1\to n-1} = [\tau^m_n - \tau^m_1, \dotsc, \tau^m_n - \tau^m_{n-1}]$.

The following lemma provides the closed-form expression of the likelihood function, enabling tractable ML detection.

\begin{lem}\label{Lem:DetectorTDoA}
The optimal detector is equivalent to a weighted summation of cross-correlations over TDoA-aligned signals:
\begin{equation}\label{eq:Lem1}
\tilde{m} = \arg\max_{m \in \{1, \dotsc, M\}} \sum_{i=2}^{N} \frac{1}{i-1} \sum_{j=1}^{i-1} \mathrm{Re} \left\{ \int \frac{\alpha^m_i}{\alpha^m_j} R^m_{i,j}(t) \, dt \right\},
\end{equation}
where $R^m_{i,j}(t) = \rv{r}_j(t - \Delta\tau^m_{i,j}) \rv{r}^*_i(t)$ and $\Delta\tau^m_{i,j} = \tau^m_i - \tau^m_j$.
\end{lem}

\begin{proof}
See Appendix \ref{Proof:DetectorTDoA}.
\end{proof}

In the noiseless case with $H_m$ true (i.e., $m = v$), the detector output simplifies as:
\begin{equation}\label{eq:noislessMequalV}
	\begin{aligned}
&\textup{Re}\Big\{\int\frac{\alpha^{m}_{i}}{\alpha^{m}_{j}}R^m_{i,j}(t)\,{d}t\Big\}\\&=\textup{Re}\Big\{\int\frac{\alpha^m_i}{\alpha^m_j} \alpha^m_j s(t-\tau^m_j-(\tau^m_i-\tau^m_j)) \left(\alpha^m_i\right)^{*}
\\&\qquad\qquad\qquad\qquad s^{*}(t-\tau^m_i)\,{d}t\Big\}\\
&=\textup{Re}\Big\{\int |\alpha^m_i|^2s(t-\tau^m_i-\tau^m_j-(\tau^m_i-\tau^m_j))s^{*}(t)\,{d}t\Big\}
\\&\overset{(a)}{=} \textup{Re}\Big\{\int |\rho|^2s(t)s^{*}(t)\,{d}t\Big\}=\rho^2\textup{Re}\Big\{\int \bar{s}(t)\bar{s}^{*}(t)\,{d}t\Big\}\\&=\rho^2E_{\mathrm s}.
  \end{aligned}
\end{equation}
\normalsize
The equality in (a) follows from the fact that
$$
|\alpha^{m}_i|^2 = \left|\rho\, \mathrm{e}^{-\jmath 2\pi f_c \tau_n^v} \, \mathrm{e}^{\jmath 2\pi \Delta f \tau_n^v}\right|^2 = \rho^2.
$$
Conversely, under $H_m$ false ($m \ne v$), the integrand becomes misaligned in time:
\begin{equation}\label{eq:noislessMnotqV}
	\begin{aligned}
		&\textup{Re}\Big\{\int\frac{\alpha^m_i}{\alpha^m_j}R^m_{i,j}(t)\,{d}t\Big\}\\
		&=\textup{Re}\Big\{\int\frac{\alpha^m_i}{\alpha^m_j} \alpha^v_js(t-\tau^v_j-(\tau^m_i-\tau^m_j))(\alpha^v_i)^{*}s^{*}(t-\tau^v_i)\,{d}t\Big\}\\
		&=\textup{Re}\Big\{\frac{\alpha^m_i}{\alpha^m_j}\alpha^v_j \left(\alpha^v_i\right)^{*} \Big\}\int s(t-\Delta\tau^v_{i,j}+\Delta\tau^m_{i,j})s^{*}(t)\,{d}t\\&\leq \alpha^2 \int s(t-\Delta\tau^v_{i,j}+\Delta\tau^m_{i,j})s^{*}(t)\,{d}t,
	\end{aligned}	
\end{equation}
which approaches zero when $\Delta\tau^v_{i,j}-\Delta\tau^m_{i,j}$ is large and is independent of the phase noise.

\begin{remark}
The weighting factor $\alpha^m_i/\alpha^m_j$ does not require CSI or phase noise estimation. It simplifies to
\begin{equation}
	\begin{aligned}
		 \frac{\alpha^m_i}{\alpha^m_j} \triangleq\frac{\rho\,\text{e}^{-\jmath 2\pi f_c\tau^m_i}}{\rho\,\text{e}^{-\jmath 2\pi f_c\tau^m_j}}=\text{e}^{-\jmath 2\pi f_c(\tau^m_i-\tau^m_j)}.
	\end{aligned}
\end{equation}
which can be computed from the known TDoA. Alternatively, the absolute value of the correlation terms can be used.
\end{remark}

\section{Direction vs. Channel Coherence Time}\label{Sec:PostionCoherenceTime}

\subsection{Conceptual Framework}\label{Sec:Motivation}

In mobile communication systems, the rate at which the channel impulse response varies is quantified by the CCT, which scales proportionally with the wavelength $\lambda$ and inversely proportional with the user velocity $v$. Together with the phase noise, the CCT fundamentally bounds the reliability of conventional amplitude/phase modulation schemes and SM variants such as Space Shift, which depend on coherent detection using CSI.

In contrast, DSK relies on the stability of the DoA (or equivalently, TDoA) rather than the full CSI. Empirical evidence and prior work (e.g., \cite{Mohaned}) suggest that the DoA varies at a significantly slower rate than the channel coefficients, particularly in the case of road side units. This observation underpins the robustness of DSK to mobility-induced variations.

To generalize this intuition, we analyze the temporal stability of the TDoA under arbitrary receiver motion. Specifically, we analyze the DCT as the time interval over which the TDoA vector remains sufficiently invariant to support coherent detection in DSK. The objective of this section is to derive analytical expressions for both CCT and DCT in the presence of mobility and phase noise, and to quantify the coherence gain of DSK relative to SSK.

\subsection{Statistical Preliminaries}\label{Sec:Preliminary}

Let $\tau$ and $\tau'$ denote the propagation delays from a fixed transmitter to a mobile receiver before and after a displacement time $t_c$. The change in delay is given by
\[
\Delta\tau = \tau' - \tau \simeq \frac{v t_c}{c} \cos(\rv{\phi}),
\]
where $c$ is the speed of light and $\rv{\phi}$ is the random direction of motion, assumed uniformly distributed in $[0, 2\pi)$. While $\Delta\tau$ depends on $\rv{\phi}$, the absolute delays $\tau$ and $\tau'$ are statistically independent of $\rv{\phi}$.

Now consider the effect of phase noise. Let $\Delta f$ and $\Delta f'$ denote the frequency offsets before and after displacement. Then the composite phase rotation difference, $2\pi(\Delta f\tau - \Delta f'\tau')$, is independent of $\rv{\phi}$ when $\Delta f \neq \Delta f'$, due to the linear independence of the coefficient vectors $[1,1]$ and $[\Delta f, \Delta f']$.

For TDoA between a pair of receive antennas, we define
\[
z = (\tau_1 - \tau_2) - (\tau'_1 - \tau'_2),
\]
which depends on the displacement angle $\rv{\phi}$. However, the individual TDoAs (e.g., $\tau_1 - \tau_2$) remain independent of $\rv{\phi}$.

\subsection{Analytical Characterization}\label{Sec:MainResults}

Let $t_{\mathrm{CCT}}$ and $t_{\mathrm{DCT}}$ denote the coherence times of the channel and direction, respectively, under a wideband signal of bandwidth $B$ and a mobile user traveling at velocity $v$. Let $d$ be the typical BS-MD distance and $l$ be the antenna separation at the receiver. In the presence of random phase noise, our finding can be summarized as follows.:
\begin{align}
t_{\mathrm{CCT}} &\simeq \frac{1}{2\pi} \frac{\lambda}{v} \, J_0^{-1} \left( \frac{1}{\sqrt{2}} \sqrt{ \frac{B}{B - |\Delta f - \Delta f'|} } \right) \leq \frac{9}{16\pi} \frac{\lambda}{v}, \\
t_{\mathrm{DCT}} &\geq \frac{1}{2\pi} J_0^{-1} \left( \frac{1}{\sqrt{2}} \right) \frac{d}{v} \frac{c}{lB} = \frac{9}{16\pi} \frac{d}{v} \frac{c}{lB},
\end{align}
where $J_0^{-1}(\cdot)$ is the inverse of the zeroth-order Bessel function of the first kind. Consequently, the ratio between DCT and CCT satisfies
\begin{equation}
\frac{t_{\mathrm{DCT}}}{t_{\mathrm{CCT}}} \geq \frac{d}{\lambda} \cdot \frac{c}{lB},
\end{equation}
highlighting a potential coherence gain of several orders of magnitude, especially in high-frequency bands (e.g., mmWave or sub-THz) where $\lambda$ is small and $B$ is large. This gain enables DSK to operate with lower overhead, enhanced reliability, and improved robustness to both phase noise and user mobility.
\subsubsection{Channel Coherence Time in the Presence of Phase Noise}

To derive a closed-form expression for the CCT, we consider a baseband pulse shaped by an ideal $\sinc$ filter with period $T = 1/B$, where $B$ is the system bandwidth \cite{DavTsePraVis:B2005}. The time-domain pulse is expressed as
\begin{equation}\label{eq:sinc}
	s(t) = \sinc[\pi t/T] = \frac{\sin[\pi t/T]}{\pi t/T},
\end{equation}
whose Fourier transform $S(f)$ is a rectangular function of bandwidth $B$ and amplitude $T$.

Assume the MD moves at speed $v$ in a random direction $\rv{\varphi}$. It first receives a reference unmodulated signal $\rv x(t)$, yielding an accurate estimate of the channel coefficient. After a time displacement $t_c$, a second signal $\rv x'(t)$ is transmitted. The propagation delays of the two signals are denoted by $\rv \tau$ and $\rv \tau'$, respectively. The received signals are modeled as:
\begin{equation}
\begin{aligned}
	\rv x(t) &= \rho \mathrm{e}^{-\jmath 2\pi f_c \rv \tau} \mathrm{e}^{-\jmath 2\pi \Delta f t}s(t - \rv \tau), \\
	\rv x'(t) &= \rho' \mathrm{e}^{-\jmath 2\pi f_c \rv \tau'} \mathrm{e}^{-\jmath 2\pi \rv \Delta f' t}s(t - \rv \tau'),
\end{aligned}
\end{equation}
where $\Delta f$ and $\rv \Delta f'$ represent oscillator frequency offsets due to phase noise before and after the displacement, respectively.

The CCT quantifies the temporal stability of the channel coefficient under mobility by measuring the correlation between $\rv x(t)$ and $\rv x'(t)$. Since the channel phase is governed by the ToA (i.e., $\rho \mathrm{e}^{-\jmath 2\pi f_c \tau}$), strong correlation implies $\rv \tau' \approx \rv \tau$, indicating phase coherence. As $\rv \tau'$ deviates, coherence is lost. Formally, the CCT is the maximum $t_c$ such that the normalized correlation remains above a predefined threshold $J_{\mathrm{th}}$:
\begin{equation}\label{eq:correlFunc}
	J_{\mathrm{CCT}}(t_c) = \left| \frac{\mathbb{E}_{\rv \varphi} \left\{ \int \rv x(t)\left(\rv x'(t)\right)^* \mathrm{d}t \right\}}{ \sqrt{ \int |\rv x(t)|^2 \mathrm{d}t \int |\rv x'(t)|^2 \mathrm{d}t } } \right| \geq J_{\mathrm{th}}.
\end{equation}

The numerator simplifies as follows:
\begin{subequations}\label{eq:CsiCoherenceTime1}
		\begin{align}
	& \left|\mathbb{E}_{\rv{\varphi}}\left\{\bigintssss x(t)\left(\rv x'(t)\right)^{*}{d}t\right\}\right|
        \\&
        =\rho\rho'\Bigg|\mathbb{E}_{\rv{\varphi}}\Bigg\{\bigintssss e^{-j2\pi f_c(\rv \tau-\rv\tau')}e^{j2\pi \Delta f't_c}e^{-j2\pi(\Delta f-\Delta f')t}\nonumber\\
        &\qquad s(t-\rv \tau) \left(\rv s(t-\rv\tau')\right)^{*}{d}t\Bigg\}\Bigg|
        \\
			&=\rho\rho'\Bigg|e^{j2\pi \Delta f't_c}\mathbb{E}_{\rv{\varphi}}\Bigg\{\bigintssss e^{-j2\pi f_c(\rv \tau-\rv \tau')}s(t)\nonumber\\
            &\qquad \left(e^{j2\pi(\Delta f-\Delta f')t}\rv s(t+\rv \tau-\rv\tau')\right)^{*}{d}t\Bigg\}\Bigg|\label{eq:second}
            \\
			&=\rho\rho'\Bigg|\frac{1}{2\pi}\int^{2\pi}_{0}\int^{\frac{B}{2}}_{\frac{-B}{2}} e^{-j2\pi f_c(\rv \tau-\rv \tau')}S(f)\nonumber\\
            &\qquad\left(e^{j2\pi f(\rv \tau-\rv \tau')}\rv S(f-(\Delta f-\Delta f'))\right)^{*}{d}f d\varphi\Bigg|\label{eq:Parseval}
			\\
			&=\rho\rho'\Bigg|\frac{1}{2\pi}\int^{2\pi}_{0}\int^{\frac{B}{2}}_{\frac{-B}{2}} e^{-j2\pi(f_c+f)(\rv \tau-\rv \tau')}S(f)\nonumber\\
            &\qquad\rv S(f-(\Delta f-\Delta f')){d}f d\varphi\Bigg|
			\\ 
   &\simeq\rho\rho'\left|\frac{T^2}{2\pi}\int^{2\pi}_{0}\int^{\frac{B}{2}-|\Delta f-\Delta f'|}_{\frac{-B}{2}} e^{j2\pi \frac{t_c v}{\lambda}\cos[\varphi]}{d}f d\varphi     \right|
			\label{eq:fourth}\\&= \rho\rho'T^2\max\left(0,B-|\Delta f-\Delta f'|\right)\left|\frac{1}{\pi}\int^{\pi}_{0} e^{j2\pi \frac{t_c v}{\lambda}\cos[\varphi]} d\varphi \right|\\
   &=\rho\rho'\max\left(0,B-|\Delta f-\Delta f'|\right)\nonumber\\&\qquad\frac{1}{\pi}\left|\int_{-f_{\max}}^{f_{\max}}\frac{1}{f_{\max}\sqrt{1-\frac{z^2}{f^2_{\max}}}}\cos[2\pi z]\,{d}z\right|\\
   &=\rho\rho'T^2\max\left(0,B-|\Delta f-\Delta f'|\right) \left|J_{0}(2\pi f_{\max})\right|,
		\end{align}
\end{subequations}\normalsize
where $f_{\mathrm{max}} = \frac{t_c v}{\lambda}$ and $\lambda = \frac{c}{f_c}$. The function $J_0$ denotes the zeroth-order Bessel function of the first kind. In the above derivation, \eqref{eq:Parseval} results from applying Parseval’s identity to transition from time to frequency domain. The simplification in \eqref{eq:second} relies on the statistical independence between the displacement direction $\rv{\varphi}$ and the post-displacement oscillator offset $\Delta f' t_c$. Furthermore, the approximation $(\rv \tau - \rv \tau') \approx \frac{v t_c}{c} \cos(\rv \varphi)$ is valid for displacements small relative to the wavelength and enables the tractable form in \eqref{eq:fourth}.

As for the denominator in \eqref{eq:correlFunc}, the signal energy is given by
\begin{equation}\label{eq:energy}
\int |x(t)|^2 dt = \int_{-\frac{B}{2}}^{\frac{B}{2}} |S(f)|^2 df = B \rho^2 T^2.
\end{equation}
An identical expression holds for $\rv x'(t)$ due to symmetry in pulse shaping. Substituting the numerator and denominator into \eqref{eq:correlFunc}, and under the condition that the phase noise remains confined within the signal bandwidth (i.e., $|\Delta f - \Delta f'| < B$), the coherence function $J_{\mathrm{CCT}}(t_c)$ takes the form:
\begin{equation}\label{eq:CsiCoherenceTime}
\begin{aligned}
J_{\mathrm{CCT}}(t_c)
&= \frac{B - |\Delta f - \Delta f'|}{B} \left| \frac{1}{\pi} \int_{-f_{\max}}^{f_{\max}} \frac{\cos(2\pi z)}{f_{\max} \sqrt{1 - \frac{z^2}{f_{\max}^2}}} , dz \right| \\
&= \frac{B - |\Delta f - \Delta f'|}{B} \left| J_0(2\pi f_{\max}) \right|,
\end{aligned}
\end{equation}
where $f_{\max} = \frac{t_c v}{\lambda}$ is the maximum Doppler frequency. Notably, $J_0(2\pi f_{\max})$ corresponds to the classical channel coherence function absent phase noise, as described in \cite[Sec. II-A, (2)]{TelTse:00}. The term $$\frac{B - |\Delta f - \Delta f'|}{B} \leq 1$$ quantifies the additional decorrelation induced by phase noise. The coherence remains high when $f_c |\tau - \rv \tau'| \ll 1$, i.e., for sufficiently short displacements relative to the wavelength.

\begin{figure*}
	\centering
	\begin{subfigure}[b]{0.45\textwidth}
		\centering
		\includegraphics[scale=0.3]{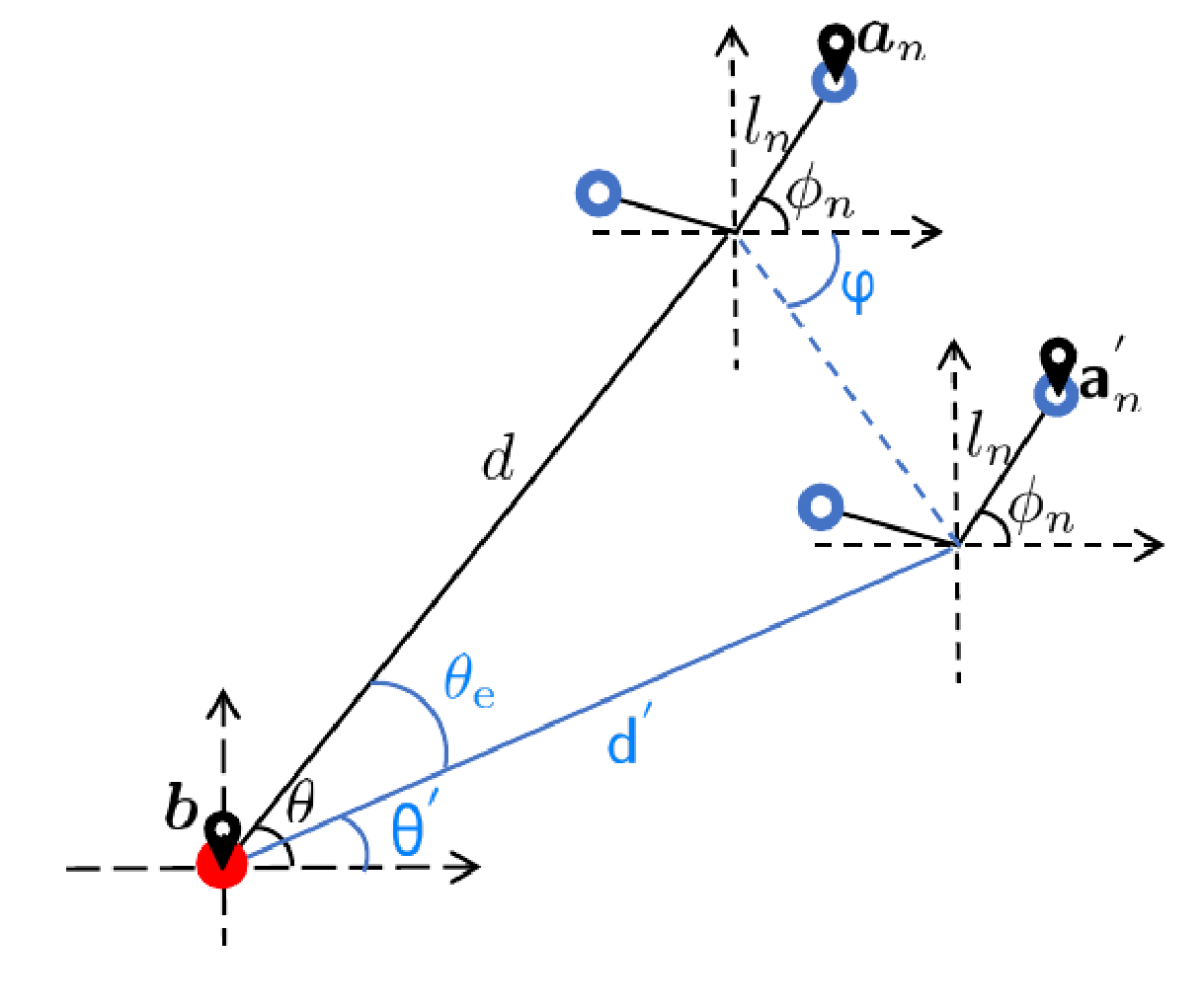}
		\caption{LoS transmissions.}
		\label{Fig:CoherenceTimeLoS}
	\end{subfigure}
	\begin{subfigure}[b]{0.45\textwidth}
		\centering
		\includegraphics[scale=0.3]{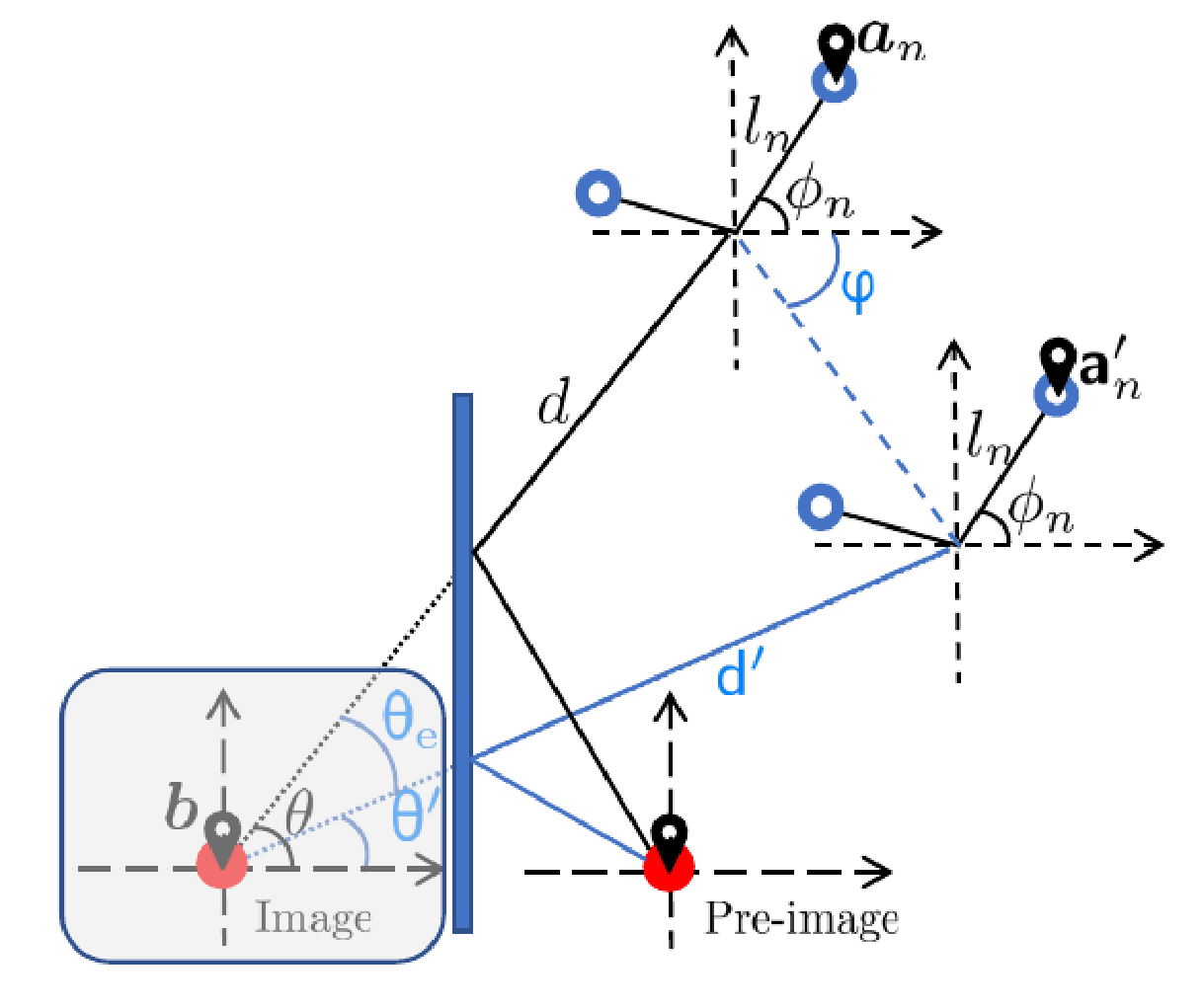}
		\caption{NLoS transmissions.}
		\label{Fig:CoherenceTimeNLoS}
	\end{subfigure}
	\caption{Geometric illustration of the propagation paths between the BS and the MD.}\label{Fig:CoherenceTime}\label{Fig:CoherenceTime}\vspace{-0.3cm}
\end{figure*}

\subsubsection{Direction Coherence Time (DCT)}

Let $\tau_2 - \tau_1$ and $\rv \tau'_2 - \rv \tau'_1$ denote TDoAs between two receive antennas, measured before and after a displacement of duration $t_c$, respectively. $\tau_2 - \tau_1$ (TDoA before displacement) is knowing to MD, but not the $\rv \tau'_2 - \rv \tau'_1$. Define $\{ \rv x'_1(t), \rv x'_2(t) \}$ as the received signals after displacement of duration $d_c$ in a random direction. The direction coherence function is defined as
\begin{equation}
\begin{aligned}
J_{\mathrm{DCT}}(t_{\mathrm{c}}) =
\left| \frac{\mathbb{E}_{\rv{\varphi}} \left\{ \int \rv x'_1(t) \left(\rv x'_2(t+\tau_2 - \tau_1)\right)^* \mathrm{d}t \right\}}{
\sqrt{ \int |\rv x'_1(t)|^2 \mathrm{d}t \int |\rv x'_2(t)|^2 \mathrm{d}t } } \right|.
\end{aligned}
\end{equation}

The numerator can be expanded as follows:
\begin{subequations}\label{eq:DCTnumerator}
\begin{align}
& \left| \mathbb{E}_{\rv{\varphi}} \left\{ \int \rv x'_1(t) \left(\rv x'_2(t+\tau_2 - \tau_1)\right)^{*} \mathrm{d}t \right\} \right| \\
&= (\rho')^2 \Bigg| \mathbb{E}_{\rv{\varphi}} \Bigg\{ \int e^{-j 2\pi f_c (\rv\tau'_1 - \rv\tau'_2)} e^{j 2\pi \Delta f' (\tau_2 - \tau_1)} \nonumber\\
& \qquad s(t - \rv \tau'_1) \left(s(t + \tau_2 - \tau_1 - \rv \tau'_2)\right)^* \mathrm{d}t \Bigg\} \Bigg| \\
&= (\rho')^2 \Bigg| e^{-j 2\pi f_c (\rv\tau'_1 - \rv\tau'_2)} e^{j 2\pi \Delta f' (\tau_2 - \tau_1)} \nonumber\\
& \qquad \mathbb{E}_{\rv{\varphi}} \left\{ \int s(t - \rv \tau'_1) \left(s(t + \tau_2 - \tau_1 - \rv \tau'_2)\right)^* \mathrm{d}t \right\} \Bigg| \label{eq:DCT_indep} \\
&=(\rho')^2\Bigg|\frac{1}{2\pi}\int^{2\pi}_{0}\int^{\frac{B}{2}}_{\frac{-B}{2}} e^{j2\pi f (\tau'_2-\tau'_1-(\tau_2-\tau_1))}S(f)\nonumber\\
&\qquad\qquad\qquad\left(S(f)\right)^{*} df d\varphi     \Bigg|\\
&= (\rho' T)^2 \left| \frac{1}{2\pi} \int_0^{2\pi} \int_{-B/2}^{B/2} e^{j 2\pi f \Delta_{\mathrm{TDoA}}} \mathrm{d}f \mathrm{d}\varphi \right|,
\end{align}
\end{subequations}

where $\Delta_{\mathrm{TDoA}} \triangleq \rv \tau'_2 - \rv \tau'_1 - (\tau_2 - \tau_1)$ denotes the deviation of the TDoA under motion. 
The result in \eqref{eq:DCT_indep} follows from the statistical independence between $\tau'_1 - \tau'_2$ and $\rv \phi$, as discussed at the end of Sec.~\ref{Sec:Preliminary}.

\begin{remark}
The numerator of the direction coherence function decays as $\Delta_{\mathrm{TDoA}}$ increases, i.e., as the deviation between the post- and pre-displacement TDoAs $(\rv \tau'_2 - \rv \tau'_1)$ and $(\tau_2 - \tau_1)$ grows. Specifically,
\begin{equation}
\int_{-B/2}^{B/2} e^{j 2\pi f \Delta_{\mathrm{TDoA}}} \mathrm{d}f = \frac{\sin(\pi B \Delta_{\mathrm{TDoA}})}{\pi \Delta_{\mathrm{TDoA}}},
\end{equation}
which vanishes as $|\Delta_{\mathrm{TDoA}}| \to \infty$. This implies that $J_{\mathrm{DCT}}(t_c)$ degrades rapidly with increasing TDoA mismatch, thereby limiting the coherence duration of DSK transmission under mobility.
\end{remark}

Substituting the energy expression from \eqref{eq:energy} for both $\rv x'_1(t)$ and $\rv x'_2(t)$, the direction coherence function simplifies to
\begin{equation}
\begin{aligned}
J_{\mathrm{DCT}}(t_c) 
= \frac{1}{B} \left| \mathbb{E}_{\rv{\varphi}} \left\{ \int_{-B/2}^{B/2} e^{j 2\pi f \Delta_{\mathrm{TDoA}}} \mathrm{d}f \right\} \right|.
\end{aligned}
\end{equation}
In contrast to the CCT derivation, where the ToA perturbation $(\rv \tau' - \tau)$ can be locally approximated as $\frac{t_c v}{c} \cos(\rv{\varphi})$ for small displacements in the order of a fraction of the wave length, this approximation fails in the DCT case. Applying the same approximation to $\Delta_{\mathrm{TDoA}}$ yields
\begin{equation}
\begin{aligned}
\Delta_{\mathrm{TDoA}} &= (\rv \tau'_2 - \tau_2) - (\rv \tau'_1 - \tau_1) \\
&= \frac{t_c v}{c} \cos(\rv{\varphi}) - \frac{t_c v}{c} \cos(\rv{\varphi}) = 0,
\end{aligned}
\end{equation}
which incorrectly suggests that $J_{\mathrm{DCT}}(t_c) = 1$ for all $t_c$. This would imply an infinite DCT, contradicting physical constraints and empirical observations. In other words, such an approximation is unrealistically favors DSK and does not accurately depicts its performance. Hence, the TDoA approximation is inapplicable in this setting and results in a significant overestimation of coherence. A more accurate kinematic model of the TDoA evolution is required to characterize the true performance of DSK under mobility.

\begin{remark}
It is important to underscore that within the CCT regime, the condition $\tau_1 \approx \rv \tau'_1$ and $\tau_2 \approx \rv \tau'_2$ implies $\tau_1 - \tau_2 \approx \rv \tau'_1 - \rv \tau'_2$. However, the reverse does not necessarily hold. This asymmetry explains why the DCT is inherently longer than the CCT, i.e., DoA remains stable even when ToAs vary. The robustness of DoA-based detection under mobility is therefore not just empirical but theoretically grounded.
\end{remark}

\begin{remark}
The proposed DCT is conceptually distinct from the classical notion of 
\emph{coherence distance} in wireless channels~\cite{Rap:96}. Coherence distance refers to the spatial scale over 
which the \emph{channel impulse response} remains significantly correlated, typically scaling with the wavelength 
$\lambda$. In contrast, DCT characterizes the temporal invariance of \emph{directional fingerprints}—specifically, 
TDoA and DoA vectors—used to infer the active transmitter in 
DSK systems. Since these fingerprints are determined by the large-scale geometry of the link rather than by 
small-scale fading, DCT scales with the BS--MD separation $d$ rather than with $\lambda$. This distinction is 
critical: a mobile device may traverse well beyond the coherence distance—causing rapid channel decorrelation—while 
the directional signatures remain stable, thereby sustaining reliable DSK detection. DCT thus complements, rather 
than replaces, classical coherence measures by introducing a geometry-driven metric of stability that explains DSK’s 
resilience to mobility and phase noise.
\end{remark}

\par The TDoA depends on the propagation path between the BS antenna and the MD. In Fig. \ref{Fig:CoherenceTime}, we provide a geometric illustration of the propagation paths for the Non-LoS (NLoS) and LoS case scenarios. To compute the TDoA for an NLoS link, we consider the image of the BS antenna over the reflection plane of the signal ray, where $\V{b_{m}}=[b_{\mathrm {x},{m}},\, b_{\mathrm {y},{m}}]$ indicates the position of the image of the BS antenna. We use the prime symbol $(.)'$ to mark the parameters related to the MD location after displacement. We use $\theta_{\mathrm e}$  denotes the phase shift $\theta-\rv\theta'$.

\begin{theorem}\label{Lem:PositionCoherenceTime}
 The direction coherence function is given by
\begin{equation}
\begin{aligned}
J_{\mathrm{DCT}}(t_{\mathrm{c}}) &= \frac{1}{\pi} \Bigg| \int_{-g_{\max}}^{g_{\max}} \left( \frac{1}{g_{\max} \sqrt{1 - \frac{z^2}{g_{\max}^2}}} + \frac{1}{\sqrt{1 - z^2}} \right) \\
&\quad \times \sinc\left[ \pi \frac{B}{c} \left( q_1 (1 - \sqrt{1 - z^2}) - q_2 z \right) \right] \mathrm{d}z \Bigg|,
\end{aligned}
\end{equation}
where
\[
\begin{cases}
q_1 = l_2 \cos(\theta - \phi_2) - l_1 \cos(\theta - \phi_1), \\
q_2 = l_2 \sin(\theta - \phi_2) - l_1 \sin(\theta - \phi_1).
\end{cases}
\]
\end{theorem}

\begin{proof}
See Appendix~\ref{Proof:PositionCoherenceTime}.
\end{proof}

In comparison to the expression of the channel coherence function  (\ref{eq:CsiCoherenceTime}), there are three main differences: the phase noise has no impact on the position coherence time; we obtain $g_{\max}$ instead of $f_{\max}$; and we get $\sinc\Big[\pi \frac{ B}{c} \Big(q_1-q_1\sqrt{1-z^2}-q_2z\Big)\Big]$ instead of $\cos(2\pi z)$. 

\par While the expression in Theorem~\ref{Lem:PositionCoherenceTime} characterizes the DCT coherence function, it does not yield direct analytical insight into the DCT scaling behavior. To bridge this gap and facilitate a insightful comparison with the CCT, we derive a tractable lower bound on $J_{\mathrm{DCT}}(t_c)$ that explicitly captures its dependence on system parameters. Without loss of generality, we consider the symmetric case $l_1 = l_2 = l$, where the receive antennas are equidistant from the propagation center. This assumption is naturally satisfied in the two-antenna case and simplifies the analysis without compromising generality.

\begin{lem}\label{Lem:LowerBoundJDCT}
The direction coherence function admits the following lower bound:
\begin{equation}
    J_{\mathrm{DCT}}(t_{\mathrm{c}}) \geq \left| J_0\left(2\pi \frac{lB}{c} g_{\max} \right) \right|,
\end{equation}
where $g_{\max} = \frac{t_{\mathrm{c}} v}{d}$ and $J_0(\cdot)$ is the zeroth-order Bessel function of the first kind.
\end{lem}

\begin{proof}
See Appendix~\ref{Proof:lemJDCTLowerBound}.
\end{proof}

\par The lower bound on $J_{\mathrm{DCT}}(t_{\mathrm c})$ mirrors the CCT expression in (\ref{eq:CsiCoherenceTime}), with $\frac{lB}{c}g_{\max} = \frac{lB}{c} \cdot \frac{t_{\mathrm c}v}{d}$ replacing $f_{\max} = \frac{t_{\mathrm c}v}{\lambda}$. Applying the standard coherence threshold $|J(t_c)|^2 \geq 0.5$ \cite{Car:87},\cite[Ch. 5]{Rap:B02}, we obtain:
\begin{equation}
	\begin{cases}
		t_{\mathrm DCT}\geq\frac{1}{2\pi}J_0^{-1}(\frac{1}{\sqrt{2}})\frac{d}{v} \frac{c}{lB}=\frac{9}{16\pi}\frac{d}{v} \frac{c}{lB}\\
		t_{\mathrm {CCT}}\simeq\frac{1}{2\pi}\frac{\lambda}{v} J_0^{-1}(\frac{1}{\sqrt{2}}\sqrt{\frac{B}{B-|\Delta f-\Delta f'|}})\leq \frac{9}{16\pi}\frac{\lambda}{v} 
	\end{cases}.
\end{equation}
Thus, the DCT scales with the BS–MD distance $d$, while the CCT scales with $\lambda$, yielding a coherence gain of
\begin{equation}
\frac{t_{\mathrm{DCT}}}{t_{\mathrm{CCT}}} \geq \frac{d}{\lambda} \cdot \frac{c}{lB}.
\end{equation}
Thus, DSK achieves multiple orders of magnitude longer coherence than that of SSK and amplitude-phase modulation under typical mmWave conditions. 

\begin{remark}Modern mobile devices are increasingly equipped with auxiliary sensing tools such as GPS, accelerometers, and digital compasses. These sensors provide side information—including velocity, heading direction, and coarse location estimates—that can be exploited to track positional changes with high reliability. Leveraging such information can substantially extend the effective Direction Coherence Time (DCT), as the receiver can compensate for mobility-induced shifts in TDoA. The present analysis excludes this auxiliary data, thereby representing a conservative lower bound on DCT performance.
\end{remark}

\section{Simulation Results}\label{Sec:SimulationResults}
In this section, we analyze the performance of DSK and compare it to SSK as a benchmark. We consider two simulation setups. In the first, we study the generic case of a circular coverage cell and evaluate the Symbol Error Rate (SER) in the presence of mobility and phase noise. In the second, a more realistic highway Road Side Unit (RSU) downlink captures close-to-reality operation of a high-mobility vehicle, implemented as an advanced version of the simulation in \cite{Mohaned}. In the RSU case, we mimic real operation by estimating the channel over a noisy link and measuring the associated overhead as a function of that estimation process. Accordingly, both the reported overhead and the channel estimates are derived from noisy pilot measurements, and the oscillator phase is modeled as a Wiener process \cite{JPhaseNoiseWiener}.

\subsection{The case of circular coverage}
We consider a circular coverage area with 100\,m radius. The BS antennas are uniformly deployed along the circumference. Transmissions occur at a carrier frequency of 30\,GHz with a bandwidth of 100\,MHz. The MD is equipped with $N = 7$ receive antennas, arranged in a circular array of radius 0.1\,m with uniform angular spacing of $2\pi/7$. We perform Monte Carlo simulations and average over multiple independent realizations. The optimal detector from Lemma~\eqref{eq:Lem1} is adopted. The receiver experiences additive white Gaussian noise. The MD moves at a constant speed of 30\,km/h in a random direction.

We benchmark the proposed DSK against conventional SSK for different modulation order $M\!\in\!\{4,16\}$ with perfect CSI for SSK and no phase noise. Fig.~\ref{Fig:SERSNR} reports SER versus SNR. For $M=4$, SSK attains the lowest SER at high SNR because it coherently exploits the full complex channel vector, whereas DSK operates in the DoA subspace only. As $M$ increases to $16$, the inter–template distances shrink and the SSK advantage diminishes; the DSK curve becomes essentially indistinguishable from SSK and can be marginally better at the highest SNRs for this geometry. Overall, DSK tracks SSK closely in the static, phase–aligned regime, with the gap narrowing as $M$ grows.
\begin{figure}
		\includegraphics[scale=0.5]{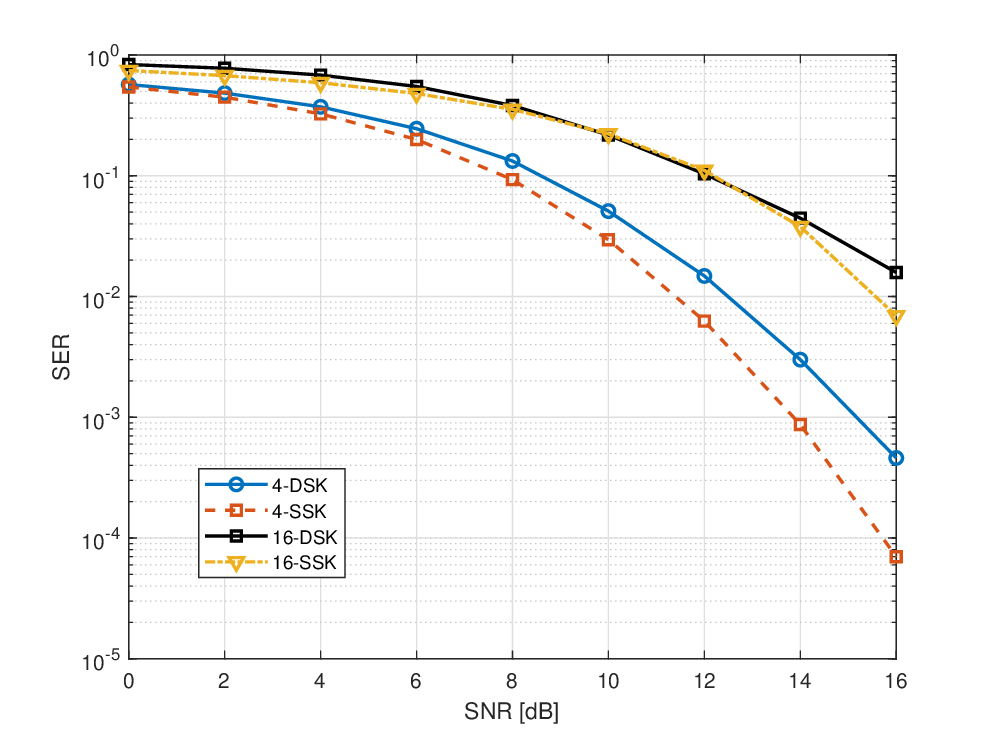}
		\caption{Symbol error rate  vs SNR (dB)}
		\label{Fig:SERSNR}
        \vspace{-0.5cm}
 \end{figure}
\begin{figure}[t]
 \centering
   \hspace*{-0.5cm} 
    \includegraphics[scale=0.44]{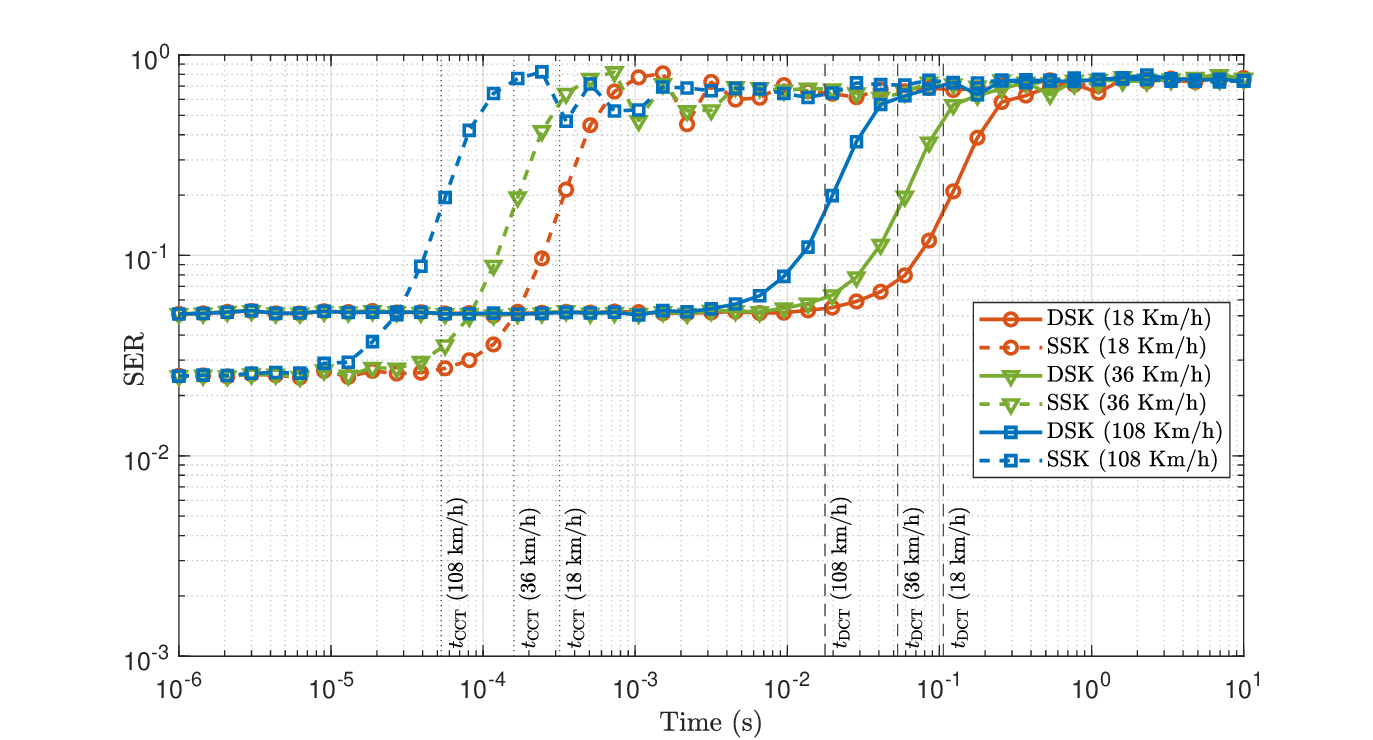}
    \caption{SER versus displacement time $t_c$ (seconds).}
    \label{fig:SERtc}
    \vspace{-0.5cm}
\end{figure}
\subsubsection{On the Effect of the Channel Coherence Time}

We next study the sensitivity of SER to the elapsed time \(t_{\mathrm{c}}\) since the last update, with all channel state references held fixed. Fig.~\ref{fig:SERtc} plots SER versus $t_{\mathrm c}$ for $M=4$ at three user speeds: $v=5$~m/s, 10~m/s, and 30~m/s (corresponding to 18, 36, and 108~km/h). As soon as $t_{\mathrm c}$ exceeds the channel coherence time
\begin{equation*}
t_{\mathrm{CCT}} \;\approx\; \frac{\lambda}{2\pi v}
\;\simeq\; 1.9\times 10^{-4}\ \text{s}\quad \text{at } f_c=30~\text{GHz},\ v=30~\text{km/h},
\end{equation*}
the SSK curves degrade precipitously, reflecting the impact of untracked CSI. In sharp contrast, DSK remains essentially flat over several decades of $t_{\mathrm c}$ and only begins to deteriorate near the directional--coherence time
\begin{equation*}
t_{\mathrm{DCT}} \;\geq\; \frac{d\lambda}{2\pi v L},
\end{equation*}
which, for the layout considered ($L=0.1$\,m, $v=30$\,km/h, $f_c=30$\,GHz), occurs beyond $\sim\!0.3$\,s. 

This means that the reference signal for coherent detection in DSK needs to be updated only at intervals that are roughly three orders of magnitude longer than the CCT, i.e., at a rate three orders lower than required for SSK. Consequently, DSK sustains coherent detection without frequent CSI updates, highlighting its inherent immunity to common--phase drift and its suitability for mobility.

\subsubsection{On the Effect of Phase Noise}
\begin{figure}[t]
\centering
   \hspace*{-0.8cm} 
		\includegraphics[scale=0.45]{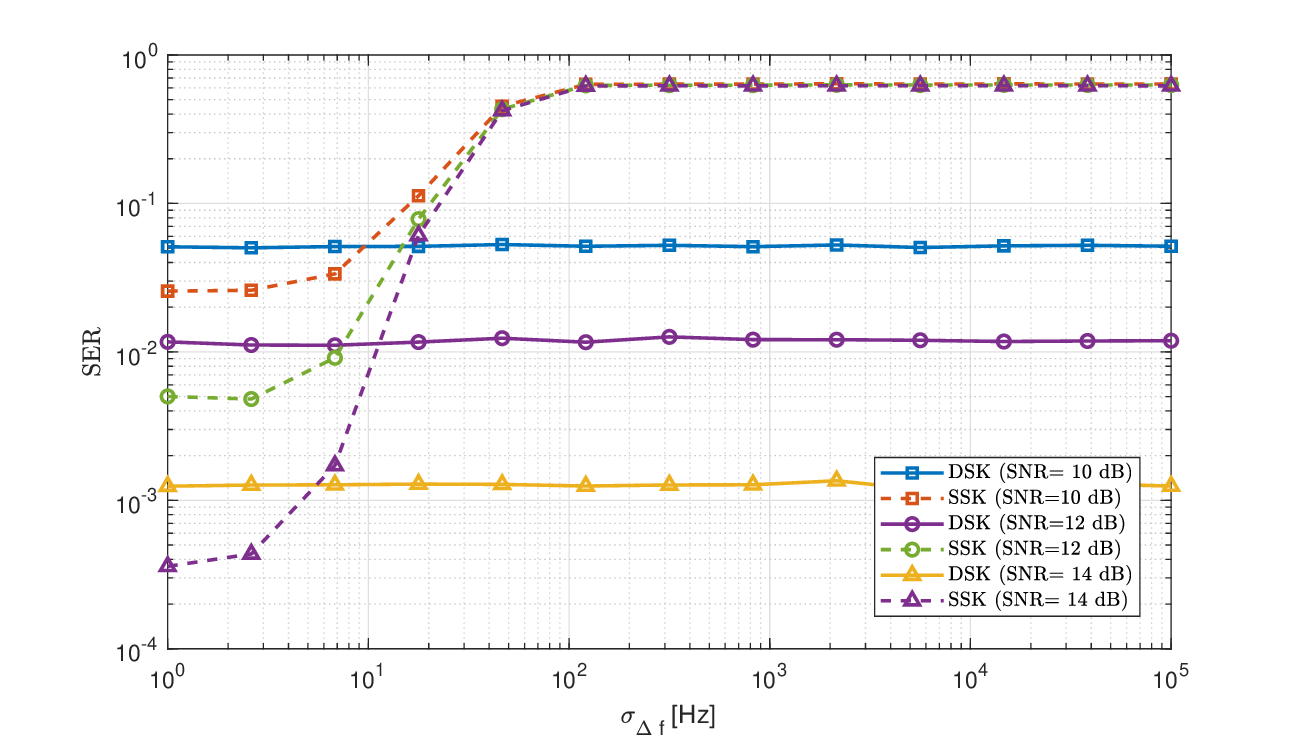}
		\caption{Symbol error rate  vs phase noise standard deviation}
		\label{Fig:ErrorRatePhaseNoise}
         \vspace{-0.4cm}
 \end{figure}
We next analyze robustness against oscillator phase noise, which is a critical impairment in high--frequency systems. 
Fig.~\ref{Fig:ErrorRatePhaseNoise} plots SER versus the phase--noise standard deviation $\sigma_{\Delta f}$ for three SNR operating points (10, 12, and 14~dB). 
The results are unambiguous: SSK suffers dramatically from phase noise, with its SER floor persisting regardless of how much the SNR is increased. This demonstrates that phase noise cannot be ``averaged out'' by higher power, i.e., it represents a fundamental bottleneck for SSK. 

In sharp contrast, the SER of DSK remains essentially flat across five decades of $\sigma_{\Delta f}$ and at all tested SNRs. This establishes DSK as intrinsically phase--noise proof: its performance is unaffected by oscillator instabilities. Together with the previous observation that DSK sustains coherent detection for update intervals three orders of magnitude longer than the channel coherence time, these results decisively show that DSK eliminates the two dominant barriers to coherent detection in mobile mmWave and sub-THz systems: 
(i) the need for extremely frequent CSI updates, and 
(ii) vulnerability to oscillator phase noise.  Hence, DSK offers a fundamentally more robust foundation for mobility than conventional SSK, enabling coherent performance under mobility and hardware impairment where SSK would collapse.
\begin{table*}[t]
\centering
\scriptsize
\setlength{\tabcolsep}{6pt}
\renewcommand{\arraystretch}{1.15}
\begin{tabular}{lrrrrrrrr}
\hline
\textbf{Pilot overhead [\%]} &
\textbf{44} & \textbf{19} & \textbf{6.42} & \textbf{2} & \textbf{0.50} & \textbf{0.18} & \textbf{0.07} & \textbf{0.02} \\
\hline
\textbf{SER (DSK)} &
$1.27\times10^{-2}$ & $1.29\times10^{-2}$ & $1.02\times10^{-2}$ & $1.29\times10^{-2}$ &
$1.29\times10^{-2}$ & $1.06\times10^{-2}$ & $9.80\times10^{-3}$ & $1.22\times10^{-2}$ \\
\textbf{SER (SSK)} &
$9.00\times10^{-4}$ & $3.76\times10^{-2}$ & $2.50\times10^{-1}$ & $3.17\times10^{-1}$ &
$3.29\times10^{-1}$ & $3.28\times10^{-1}$ & $3.30\times10^{-1}$ & $3.41\times10^{-1}$ \\
\hline
\end{tabular}
\caption{SER versus overhead percentage.}
\label{tab:ser_vs_overhead_repr}
\vspace{-0.5cm}
\end{table*}

\subsection{The case of RSU with overhead estimation }

We consider a highway RSU downlink intended to mirror deployment practice and analyze the efficacy of DSK under high mobility. The road centerline is modeled as a piecewise-linear trajectory; RSUs are placed every \(100\,\mathrm{m}\) with a fixed lateral offset of \(10\,\mathrm{m}\). The mobile device carries a \(N{=}5\) element Uniform Linear Array (ULA) with \(0.5\lambda\) spacing. Channel references are obtained via pilot-aided estimation over a noisy channel: at each segment entry and then periodically, the receiver collects \(N_p{=}4\) pilots per transmitter and forms references by averaging the received pilot vectors; for SSK the reference is the averaged complex \(N\)-vector, whereas for DSK it is a phase-only, amplitude-invariant feature obtained by de-rotating to the first antenna, normalizing magnitudes element-wise, and \(\ell_2\)-normalizing the \((N{-}1)\)-length tail. Updates are attempted every \(U=\lceil T_{\mathrm{upd}}/T_s\rceil\) symbols with symbol period \(T_s=1\,\mu\mathrm{s}\). The pilot-overhead ratio is reported  empirically as \(\tfrac{2N_p}{U+2N_p}\) as a function of \(T_{\mathrm{upd}}\) and $T_s$. Unlike the circular-cell baseline that normalizes SNR at the receiver, here we fix  transmit power \(P_{\mathrm{tx}}\) and per-antenna noise variance \(\sigma^2=10^{-12}\,\mathrm{W}\). Propagation follows free-space path loss, yielding a spatially varying SNR along the road (we also plot the SNR profile for reference). Oscillator phase noise is modeled as a common-phase-error Wiener (random-walk) process that updates every symbol period and is characterized by the standard deviation of the oscillator linewidth denoted by $\sigma_{\Delta f}$ \cite{JPhaseNoiseWiener}. 

\begin{figure}[h]
 \centering
   \hspace*{-0.8cm} 
    \includegraphics[scale=0.45]{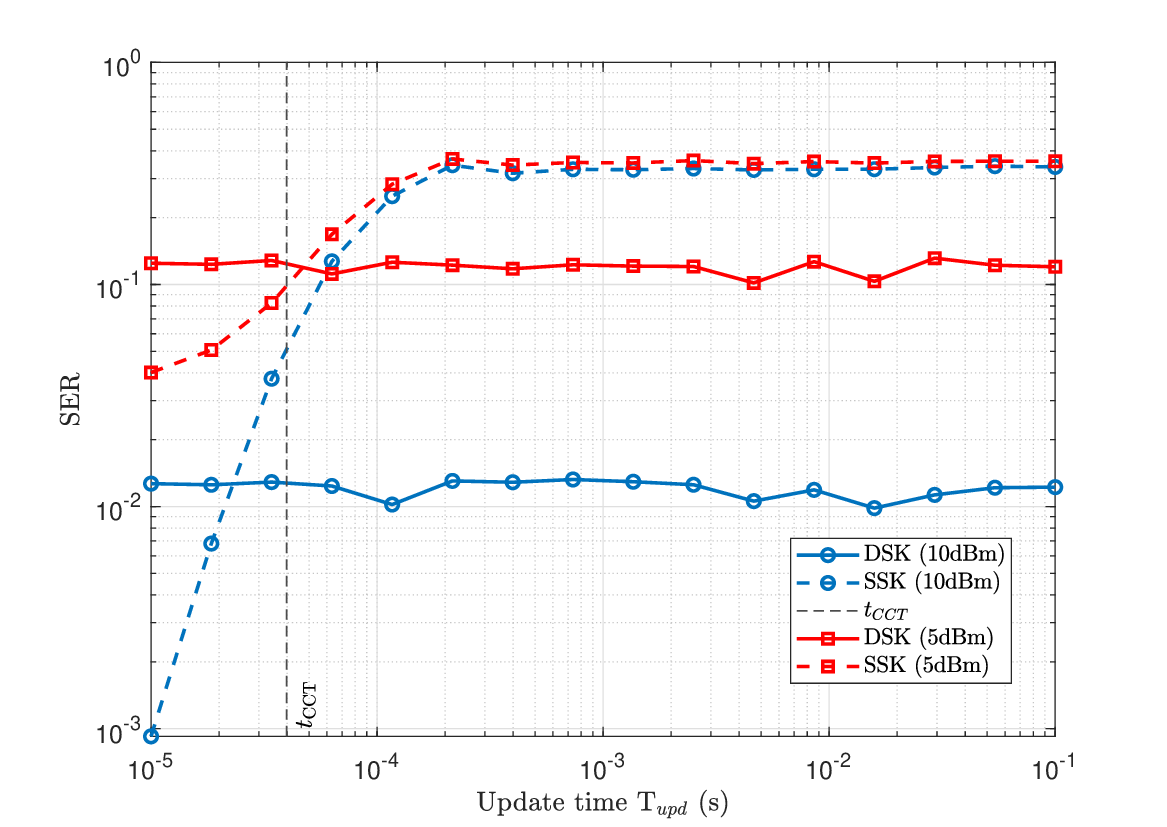}
    \caption{SER as function of the pilot symbol periods (s).}
    \label{fig:SERvsTupd}
    \vspace{-0.5cm}
\end{figure}
In Fig.~\ref{fig:SERvsTupd} we analyze SER versus the update time \(T_{\mathrm{upd}}\) for two transmit powers (\(P_{\mathrm{tx}}\in\{5,10\}\,\mathrm{dBm}\)). At a 30\,GHz carrier, the free-space baseband gain is \(|a|=\lambda/(4\pi d)\), so the per-antenna received power is \(P_{\mathrm{rx}}=P_{\mathrm{tx}}|a|^2\). With \(P_{\mathrm{tx}}=5\,\mathrm{dBm}=3.162\times10^{-3}\,\mathrm{W}\), \(d=50\,\mathrm{m}\), and noise variance \(\sigma^2=10^{-12}\,\mathrm{W}\), we obtain
\[
\mathrm{SNR}=\frac{P_{\mathrm{rx}}}{\sigma^2}=0.801\;\;(\approx -0.96~\mathrm{dB}).
\]
With \(N\) antennas, one add \(10\log_{10}N\) dB, e.g., \(+7\) dB for \(N=5\) yields a received SNR of \(\approx 6.0\) dB at \(50\,\mathrm{m}\). Since the vehicle is moving and the distance form the RSUs changes and hence the SNR.

With \(T_s=10^{-6}\) s, the pilot symbols are refreshed every \(U=\lceil T_{\mathrm{upd}}/T_s\rceil\) symbols. The pilot overhead decreases as \(T_{\mathrm{upd}}\) grows; the corresponding overhead ratios are reported in Table~\ref{tab:ser_vs_overhead_repr}. We observe that SSK maintains low error only when the overhead is high, about \(44\%\) down to \(\sim 20\%\), and beyond this range its SER rapidly deteriorates. In contrast, DSK remains essentially flat (on the order of \(10^{-2}\)) across more than three decades of \(T_{\mathrm{upd}}\); notably, the overhead can be driven down to \(\approx 0.02\%\) with no loss for DSK. These trends align with prior reports that coherent schemes under high mobility may require $20\%$ to $30\%$ channel training overhead to preserve performance \cite{ChiVenVal:3,ChiVenVal:1}.

\begin{figure}[t]
 \centering
 \hspace*{-0.8cm}
 \includegraphics[scale=0.45]{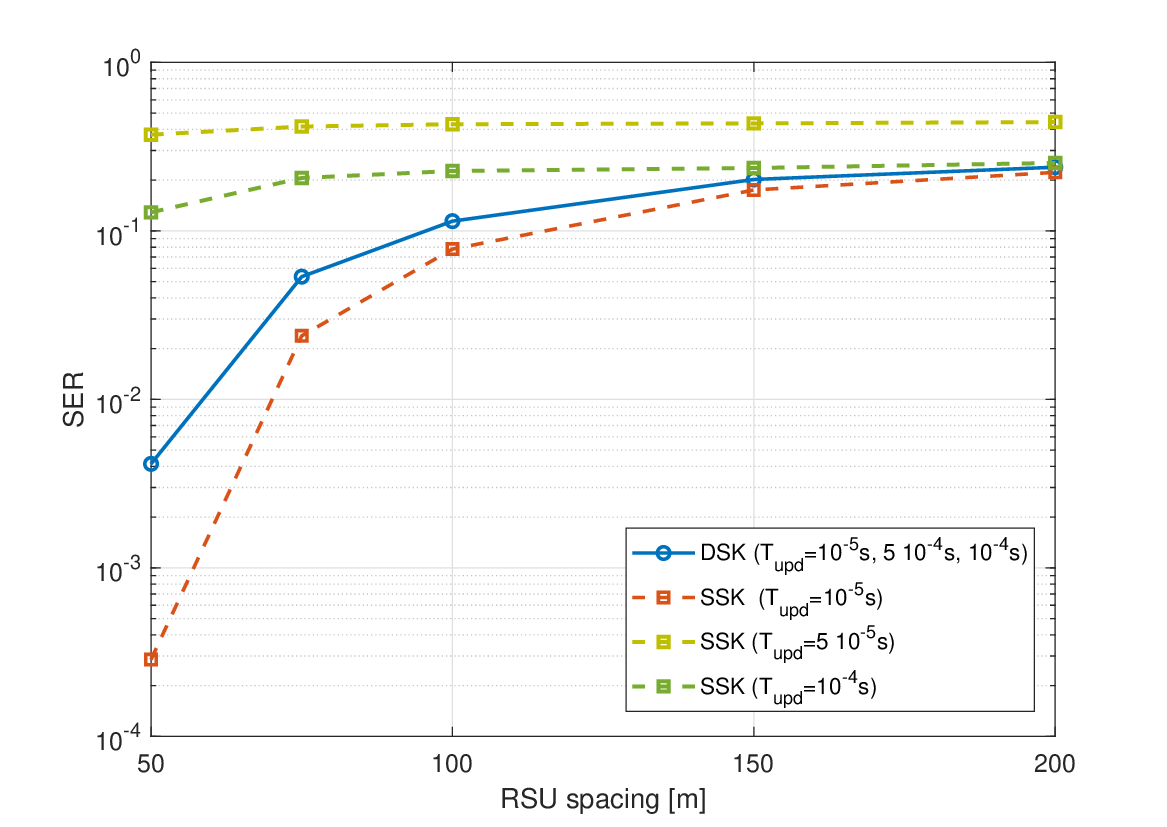}
 \caption{SER versus RSU inter-distance (meters) at fixed phase noise; \(P_{\mathrm{tx}}=12\,\mathrm{dBm}\).}
 \label{fig:RSUdistance}
  \vspace{-0.5cm}
\end{figure}

In Fig.~\ref{fig:RSUdistance} we analyze the performance as a function of the RSU inter-distance for several pilot update rates (i.e., $T_{\mathrm{upd}}$). For $P_{\mathrm{tx}}=12\,\mathrm{dBm}$ and an inter-RSU spacing of $50\,\mathrm{m}$, SSK outperforms DSK when updates are sufficiently frequent. As the spacing increases, the gap narrows. Moreover, as the pilot rate decreases (i.e., $T_{\mathrm{upd}}$ increases), DSK remains largely unaffected, whereas SSK deteriorates markedly.

\begin{figure}[h]
 \centering
 \hspace*{-0.8cm}
 \includegraphics[scale=0.45]{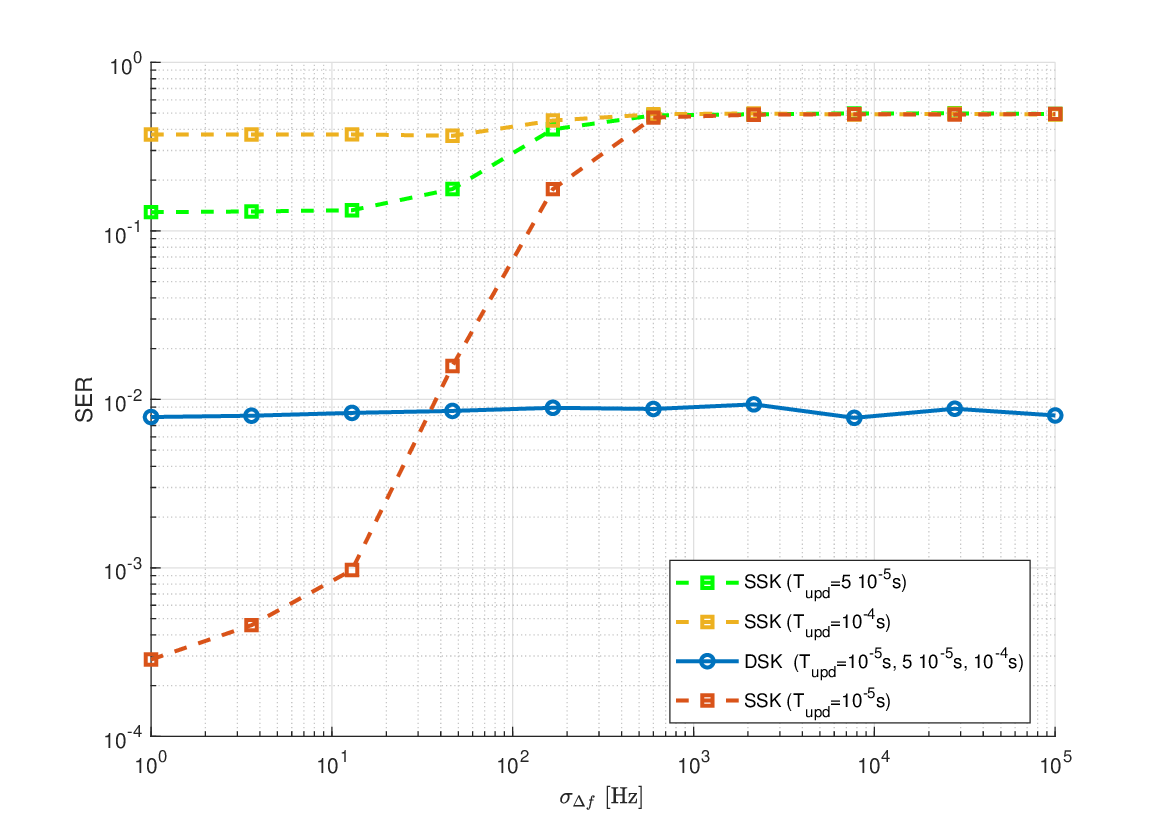}
 \caption{SER versus phase-noise level \(\sigma_{\Delta f}\) for multiple \(T_{\mathrm{upd}}\); \(P_{\mathrm{tx}}=12\,\mathrm{dBm}\), \(d=50\,\mathrm{m}\).}
 \label{fig:RSUphasenoise}
  \vspace{-0.5cm}
\end{figure}

To jointly visualize the effects of update time and phase noise, Fig.~\ref{fig:RSUphasenoise} plots SER versus the phase-noise standard deviation for multiple pilot update rates. DSK is robust to both phase noise and update time, while SSK is sensitive to increases in either. For example, once the phase-noise standard deviation exceeds $\sim\!100\,\mathrm{Hz}$ (with $T_{\mathrm{upd}}=0.01\,\mathrm{ms}$), SSK performance degrades rapidly toward the random-decision floor ($\mathrm{SER}\!\approx\!0.5$).

\section{Conclusions}\label{Sec:Conclusions}
This work presents the first analytical performance characterization of DSK, a modulation technique designed to overcome the two fundamental impairments of mmWave and sub-THz communications: oscillator phase noise and short channel coherence time. We derived the optimal detector for DSK in mobile settings and established the notion of DCT, showing that it scales with antenna separation over velocity, $d/v$, as opposed to conventional channel coherence time which scales with $\lambda/v$. This results in a coherence gain on the order of $d/\lambda$, i.e., several orders of magnitude in practical regimes.

Our findings demonstrate that DSK enables coherent detection over significantly longer time scales and exhibits inherent robustness to phase noise—requiring neither channel tracking nor phase correction. These properties substantially reduce pilot overhead and improve spectral efficiency in high-mobility, high-frequency environments. The analysis provides a theoretical foundation for DSK and establishes it as a promising candidate for mobility-resilient and hardware-efficient mmWave communication. 

The proposed approach effectively mitigates phase noise and rapid channel variations but requires a shift in system design: instead of centrally located BS antennas, transmitting antennas should be placed around the coverage perimeter to maximize separation and symbol distinguishability, calling for further study of system planning in dense deployments.

\appendices

\section{Proof Of Lemma \ref{Lem:DetectorTDoA}}\label{Proof:DetectorTDoA}
The expressions in (\ref{eq:ML}) and (\ref{eq:DetectorMLTDoA}) together define the form of the optimal TDoA detector. Further developing the detector expression yields to an expression for $\tilde{m}$ as
\newcommand\independent{\protect\mathpalette{\protect\independenT}{\perp}}
\def\independenT#1#2{\mathrel{\rlap{$#1#2$}\,\mkern2mu{#1#2}}}
\begin{align} \small\label{eq:MLanalysis}
&\underset{{m}\in\{{1} \dotsc {M}\}}{\operatorname{argmax}} \prod_{n=1}^{N}f_{\rv{r}_n|\RV{\Delta}\RV{\tau}, \rv r_{1\to n-1}}(r_{n}|\V{\Delta}\V{\tau}^{m}_{n,1\to n-1},\V{r}_{1\to n-1})\nonumber\\
 &=\underset{{m}\in\{{1} \dotsc {M}\}}{\operatorname{argmax}}\log\Big\{\underbracee{\independent{m}}{f_{\rv{r}_1}(r_1)}\Big\}\nonumber\\
 &\qquad+\sum_{n=2}^{N}\log\big\{f_{\rv{r}_n|\RV{\Delta}\RV{\tau}, \rv r_{1\to n-1}}(r_{n}|\V{\Delta}\V{\tau}^{m}_{n,1\to n-1},\V{r}_{1\to n-1})\big\}\nonumber\\
  =&\underset{{m}\in\{{1} .. {M}\}}{\operatorname{argmax}}\sum_{n=2}^{N}\log\left\{f_{\rv{r}_n|\RV{\Delta}\RV{\tau}, \rv r_{1\to n-1}}(r_{n}|\V{\Delta}\V{\tau}^{m}_{n,1\to n-1},\V{r}_{1\to n-1})\right\}
 \end{align}
\normalsize
Here, the symbol $\independent$ indicates that $f_{\rv{r}_1}$ is not dependent of the index $m$. In the above derivations, we use the fact that maximizing a positive likelihood function is equivalent to maximizing its logarithmic. In the following analysis, we start by deriving the closed-form expression for the first term in the sum $f_{\rv{r}_2|\RV{\Delta}\RV{\tau}, \rv r_{1}}(r_2|\Delta\tau^m_{2,1},{r}_1)$.

\par Assuming that the hypothesis $H_m$ is true, we have
\[\begin{cases}\small
\rv{r}_1(t)\sim \mathcal{CN}(\alpha^m_1 s(t-t_0-\tau^m_1),\sigma^2 )\\
\rv{r}_2(t)\sim \mathcal{CN}(\alpha^m_2 s(t-t_0-\tau^m_2),\sigma^2 )
\end{cases}.\]\normalsize

 These imply that
\begin{equation}\small\label{eq:mean}
\begin{aligned}
\mathbb{E}_{\rv w_2}\{\rv r_2(t)\}&=\alpha_2 s(t-t_0-\tau^m_2)\\
&=\frac{\alpha^m_2}{\alpha^m_1}\Big\{\rv{r}_1(t-(\underbrace{\tau^m_2-\tau^m_1}_{\Delta\tau^m_{2,1}}))-\rv w_1(t-(\underbrace{\tau^m_2-\tau^m_1}_{\Delta\tau^m_{2,1}})\Big\}.
\end{aligned}\normalsize
\end{equation}
Accordingly, we have 
\begin{equation}\small\label{eq:likelihood}
\begin{aligned}
&f_{\rv{r}_2|\RV{\Delta}\RV{\tau}, \rv r_{1}}(r_2|\Delta\tau^m_{2,1},{r}_1)
\\&=\int f_{\rv{r}_2|\RV{\Delta}\RV{\tau}, \rv r_{1}, \rv w_1}(r_2|\Delta\tau^m_{2,1},{r}_1, w_1) f_{\rv w_1}(w_1) \, d w_1
\\& \coloneqq
\mathbb{E}_{\rv w_1}\{f_{\rv{r}_2|\RV{\Delta}\RV{\tau}, \rv r_{1}, \rv w_1}(r_2|\Delta\tau^m_{2,1},{r}_1, \rv w_1)\}\\
&=\mathbb{E}_{\rv w_1}\Big\{\frac{1}{\sqrt{2\pi\sigma^2}}\exp\Big[-\frac{1}{2\sigma^2}\int|r_2(t)-\mathbb{E_{\rv{w_2}}}\{\rv r_2(t)\}|^2 {d}t\Big]\Big\}\\
&=\mathbb{E}_{\rv w_1}\Big\{\frac{1}{\sqrt{2\pi\sigma^2}}\exp\Big[-\frac{1}{2\sigma^2}\int|r_2(t)-\frac{\alpha^m_2}{\alpha^m_1}
\\&\qquad\qquad\qquad \left( r_1(t-\Delta\tau^m_{2,1})-\rv w_1(t-\Delta\tau^m_{2,1})\right)|^2{d}t\Big]\Big\}.
\end{aligned}
\end{equation}\normalsize
Expanding the integral gives the sum of four terms:
\begin{flalign}\small
&\int|r_2(t)|^2{d}t+\int\left|\frac{\alpha_2}{\alpha_1}r_1(t-\Delta\tau^m_{2,1})\right|^2{d}t\nonumber\\&\qquad\qquad =\underbrace{\int|r_2(t)|^2{d}t+\int\left|\frac{\alpha^m_2}{\alpha^m_1}{r}_1(t)\right|^2{d}t}_{\circledone},
\end{flalign}
\begin{flalign}
	&\underbrace{\int- 2\,\text{Re}\Big\{\frac{\alpha^m_2}{\alpha^m_1}r_1(t-\Delta\tau^m_{2,1})r^{*}_{2}(t)\Big\}\,{d}t}_{\circledtwo},
\end{flalign}
\begin{flalign}
&\int 2\,\text{Re}\Big\{\frac{\alpha^m_2}{\alpha^m_1}\rv w_1(t-\Delta\tau^m_{2,1})\,\Big[
r_2(t)-\frac{\alpha^m_2}{\alpha^m_1}\rv{r}_1(t)
\Big]^{*}\Big\}\,{d}t=\nonumber
\\&\underbrace{\int 2\,\text{Re}\Big\{\frac{\alpha^m_2}{\alpha^m_1} \rv w_1(t-\Delta\tau^m_{2,1})\Big[w_2(t)-\frac{\alpha^m_2}{\alpha^m_1}\rv w_1(t-\Delta\tau^m_{2,1})\Big]^{*}\Big\}}_{\circledthree}
\end{flalign}

\begin{flalign}\label{eq:four}
	\text{and}&\int\left|\frac{\alpha^m_2}{\alpha^m_1}\rv w_1(t-\Delta\tau^m_{2,1})\right|^2{d}t=\underbrace{\int\left|\frac{\alpha^m_2}{\alpha^m_1}\rv w_1(t)\right|^2{d}t}_{\circledfour}.
\end{flalign}
\normalsize
The terms $\circledfour$ does not depend on $\Delta_{2,1}^m$. Therefore, it can be omitted later on from the detector expression. Moreover, terms $\circledone$ and $\circledtwo$ are independent of $\circledthree$ and $\circledfour$, and this will be of later use in the derivation of (\ref{eq:MLfinal}). The fact that $\{\rv w_1(t),t\in \mathbb{R}\}$ is stationary process gives

 \small\[ \begin{aligned} &\mathbb{E}_{\rv w_1}\left\{\frac{1}{\sqrt{2\pi\sigma^2}}\exp\Big[-\frac{1}{2\sigma^2} \circledthree \Big]\right\}
 \\
 	&=\mathbb{E}_{\rv w_1}\Bigg\{\frac{1}{\sqrt{2\pi\sigma^2}}\exp\Big[-\frac{1}{2\sigma^2}\int 	2\text{Re}\Big\{\frac{\alpha^m_2}{\alpha^m_1} \rv w_1(t-\Delta\tau^m_{2,1})\\
  &\qquad\qquad\qquad\qquad\Big(w_2(t)-\frac{\alpha^m_2}{\alpha^m_1}\rv w_1(t-\Delta\tau^m_{2,1})\Big)^{*}\Big\} \,{d}t\Big]\Bigg\}
 \\
 		&=\mathbb{E}_{\rv w_1}\Bigg\{\frac{1}{\sqrt{2\pi\sigma^2}}\exp\Big[-\frac{1}{2\sigma^2}\int 2\text{Re}\Big\{\frac{\alpha^m_2}{\alpha^m_1} \rv w_1(t)\\
   &\qquad\qquad\qquad\Big(w_2(t)-\frac{\alpha^m_2}{\alpha^m_1}\rv w_1(t)\Big)^{*}\Big\}\,dt\Big]\Bigg\},
 \end{aligned}\]
\normalsize
which is independent of $\Delta_{2,1}^m$ and hence can be omitted from the detector expression. 

\par Considering the above results regarding $\circledthree$ and $\circledfour$ and incorporating the expressions of the four terms in (\ref{eq:likelihood}) and  (\ref{eq:MLanalysis}) gives
\begin{align}\small\label{eq:MLfinal}
&\underset{{m}\in\{{1} \dotsc {M}\}}{\operatorname{argmax}}\log\left\{f_{\rv{r}_2|\RV{\Delta}\RV{\tau}, \rv r_{1}}(r_{n}|\V{\Delta}\V{\tau}^{m}_{2,1},\V{r}_{1})\right\}\nonumber\\
&\qquad+\sum_{n=3}^{N}\log\left\{f_{\rv{r}_n|\RV{\Delta}\RV{\tau}, \rv r_{1\to n-1}}(r_{n}|\V{\Delta}\V{\tau}^{m}_{n,1\to n-1},\V{r}_{1\to n-1})\right\}\nonumber\\
&=\underset{{m}\in\{{1} \dotsc {M}\}}{\operatorname{argmax}}\log\mathbb{E}_{\rv w_1}\Big\{\exp\Big[-\frac{1}{2\sigma^2}(\circledone+\circledtwo+\circledthree+\circledfour)\Big]\Big\}\nonumber\\
&\qquad+\sum_{n=3}^{N}\log\left\{f_{\rv{r}_n|\RV{\Delta}\RV{\tau}, \rv r_{1\to n-1}}(r_{n}|\V{\Delta}\V{\tau}^{m}_{n,1\to n-1},\V{r}_{1\to n-1})\right\}\nonumber\\
&\overset{(c)}{=}\underset{{m}\in\{{1} \dotsc {M}\}}{\operatorname{argmax}}\log\mathbb{E}_{\rv w_1}\Big\{\exp\Big[-\frac{1}{2\sigma^2}(\circledone+\circledtwo)\Big]\Big\}\nonumber\\
&\qquad+\log\underbracee{\independent {m}}{\mathbb{E}_{\rv w_1}\Big\{\exp \Big[-\frac{1}{2\sigma^2}(\circledthree+\circledfour)\Big]\Big\}}\nonumber\\&
\qquad+\sum_{n=3}^{N}\log\left\{f_{\rv{r}_n|\RV{\Delta}\RV{\tau}, \rv r_{1\to n-1}}(r_{n}|\V{\Delta}\V{\tau}^{m}_{n,1\to n-1},\V{r}_{1\to n-1})\right\}\nonumber\\
& \overset{(d)}{=}\underset{{m}\in\{{1} \dotsc {M}\}}{\operatorname{argmax}}\log\Big\{\exp\Big[-\frac{1}{2\sigma^2}(\circledone+\circledtwo)\Big]\Big\}\nonumber\\
&\qquad+\sum_{n=3}^{N}\log\left\{f_{\rv{r}_n|\RV{\Delta}\RV{\tau}, \rv r_{1\to n-1}}(r_{n}|\V{\Delta}\V{\tau}^{m}_{n,1\to n-1},\V{r}_{1\to n-1})\right\}\nonumber\\&
=\underset{{m}\in\{{1} \dotsc {M}\}}{\operatorname{argmax}}-\frac{1}{2\sigma^2}\Bigg\{\underbracee{\independent{m}}{\int|r_2(t)|^2{d}t}+\underbracee{\independent{m}}{\int\left|\frac{\alpha^m_2}{\alpha^m_1}\rv{r}_1(t)\right|^2{d}t}\nonumber\\
&\qquad -\int 2\,\text{Re}\Big\{\frac{\alpha^m_2}{\alpha^m_1}r_1(t-\Delta\tau^m_{2,1})r^{*}_{2}(t)\Big\}\,{d}t\Bigg\}\nonumber\\&
\qquad+\sum_{n=3}^{N}\log\left\{f_{\rv{r}_n|\RV{\Delta}\RV{\tau}, \rv r_{1\to n-1}}(r_{n}|\V{\Delta}\V{\tau}^{m}_{n,1\to n-1},\V{r}_{1\to n-1})\right\}\nonumber\\
&=\underset{{m}\in\{{1} \dotsc {M}\}}{\operatorname{argmax}}\frac{1}{\sigma^2}\int\text{Re}\Big\{\frac{\alpha^m_2}{\alpha^m_1}r_1(t-\Delta\tau^m_{2,1})r^{*}_{2}(t)\Big\}\,{d}t
\nonumber
\\ 
&\, +\sum_{n=3}^{N}\log\left\{f_{\rv{r}_n|\RV{\Delta}\RV{\tau}, \rv r_{1\to n-1}}(r_{n}|\V{\Delta}\V{\tau}^{m}_{n,1\to n-1},\V{r}_{1\to n-1})\right\}
\end{align}
\normalsize
The equality in $(c)$ stems from the fact that $\circledone$ and $\circledtwo$ are independent from $\circledthree$ and $\circledfour$. Moreover, the passage $(d)$ is the results of $\circledthree$ and $\circledfour$ being independent from $\Delta_{2,1}^m$ post applying the expectation operator. In the next, we extend our proof for the $(n-1)$th likelihood term: 
$$\log\left\{f_{\rv{r}_n|\RV{\Delta}\RV{\tau}, \rv r_{1\to n-1}}(r_{n}|\V{\Delta}\V{\tau}^{m}_{n,1\to n-1},\V{r}_{1\to n-1})\right\}.$$ For the next, we proceed in the same way as the derivation of the first term. We hence start by writing the mean of $\rv r_{n}$ as a function of the other received signals in $\RV r_{1\to n}$ and TDoAs vector $\V\Delta\V \tau^m_{n, 1\to n-1}$. By analogy to (\ref{eq:mean}), we have
\begin{equation}\small\label{eq:meani}
	\begin{aligned}
		\mathbb{E}_{\rv w_n}\{\rv r_n(t)\}&=\alpha^m_n s(t-t_0-\tau^m_n)\\
&
  =\frac{1}{n-1}\sum_{i=1}^{n-1}\frac{\alpha^m_n}{\alpha^m_i}\Big\{\rv r_{i}(t-\Delta\tau^m_{n,i})-\rv w_{i}(t-\Delta\tau^m_{n,i})\Big\}.
	\end{aligned}
\end{equation}\normalsize
The above expression has similar form to the one in (\ref{eq:mean}). In fact, it can be obtained by substituting in (\ref{eq:mean}) the terms $\frac{\alpha^m_2}{\alpha^m_1}\rv{r}_1(t-\Delta\tau^m_{2,1})$ and $\rv w_1(t-\Delta\tau^m_{2,1})$, respectively, by $\frac{1}{n-1}\sum_{i=1}^{n-1}\frac{\alpha^m_i}{\alpha^m_j}\rv r_{i}(t-\Delta\tau^m_{n,i})$ and $\frac{1}{n-1}\sum_{i=1}^{n-1}\rv w_{i}(t-\Delta\tau^m_{n,i})$. The same analogy holds for the rest of the derivations, starting from the expression (\ref{eq:mean}) up to the expression (\ref{eq:MLfinal}). Following the exact same steps for each of the terms in the detector expression, we obtain a closed-form expression of the detector:
\begin{equation}\small
\begin{aligned}
\tilde{m}=&\underset{{m}\in\{{1} \dotsc {M}\}}{\operatorname{argmax}}\sum_{n=2}^{N}\frac{1}{n-1}\sum_{i=1}^{n-1}\text{Re}\Big\{\int
\frac{\alpha^m_n}{\alpha^m_i}r_{n}(t-\Delta\tau^m_{n,i})r^{*}_{i}(t)\,{d}t\Big\}.
\end{aligned}
\end{equation}
\normalsize

\section{Proof Of Theorem\ref{Lem:PositionCoherenceTime}}\label{Proof:PositionCoherenceTime}
 The difference of TDoAs can be written as
\begin{equation}\small
	\begin{aligned}
&\tau_2-\tau_1-(\rv\tau'_2-\rv\tau'_1)\\&=\frac{1}{c}\Big(l_2\cos[\theta-\phi_2]-l_1\cos[\theta-\phi_1]\\
&\qquad\qquad
-(l_2\cos[\rv\theta'-\phi_2]-l_1\cos[\rv\theta'-\phi_1])\Big)\\
	&=\frac{1}{c}\Big(l_2\cos[\theta-\phi_2]-l_1\cos[\theta-\phi_1]\\
&\qquad\qquad-(l_2\cos[\theta-\rv \theta_{\mathrm e}-\phi_2]-l_1\cos[\theta-\rv \theta_{\mathrm e}-\phi_1])\Big).
	\end{aligned}
\end{equation}\normalsize
We have $$\cos[\theta-\rv \theta_{\mathrm e}-\phi_1]=\cos[\rv \theta_{\mathrm e}]\cos[\theta-\phi_1]+\sin[\rv \theta_{\mathrm e}]\sin[\theta-\phi_1].$$
Incorporating this equality in the expression of the TDoAs-difference gives
\begin{equation}\small
	\begin{aligned}
		&\tau_2-\tau_1-(\rv\tau'_2-\rv\tau'_1)=\frac{1}{c}\Big(q_1-q_1\cos[\rv \theta_{\mathrm e}]-q_2\sin[\rv \theta_{\mathrm e}]\Big),
	\end{aligned}
\end{equation}
\vspace{0.2cm}
\normalsize
where $$\begin{cases} q_1=l_2\cos[\theta-\phi_2]-l_1\cos[\theta-\phi_1]\\ q_2=l_2\sin[\theta-\phi_2]-l_1\sin[\theta-\phi_1]\end{cases}.$$
To derive the expression of the coherence function, we need to integrate over all values of $\rv\varphi$. Meanwhile, $\rv \Delta\rv \tau$ is a function of $\rv \theta_{\mathrm e}$. In the next, we express $\rv \theta_{\mathrm e}$ as a function of $\rv \varphi$ and then perform an integration by substitution. From Fig. \ref{Fig:CoherenceTime}, we have
$$\rv d'\exp[\jmath\rv\theta']=d\exp[\jmath\theta]+t_{\mathrm c}v\exp[\jmath\rv\varphi].$$ Multiplying each of the equality sides by $\exp[-\jmath\rv\theta']$ gives
$$\rv d'=d\exp[\jmath\rv \theta_{\mathrm e}]+t_{\mathrm c}v\exp[\jmath(\rv\varphi-\rv\theta')].$$
A key observation is that the imaginary part of $\rv d'$ is equal to $0$ and, hence, is the imaginary part of the right side of the equation. This implies $d\sin[\rv \theta_{\mathrm e}]+t_{\mathrm c}v\sin[\rv \varphi-\underbrace{\rv\theta'}_{\theta-\rv \theta_{\mathrm e}}]=0$, which yields
$$\varphi=\arcsin\left(-\frac{d}{t_{\mathrm c}v}\sin[\rv \theta_{\mathrm e}]\right)-\rv \theta_{\mathrm e}+\theta.$$
Moreover, we have
\begin{equation}\small
	\begin{aligned}
		f_{\rv \theta_{\mathrm e}}(\theta_{\mathrm e})&=f_{\rv \varphi}(\varphi)\left|\frac{\partial \varphi}{\partial \theta_{\mathrm e}}\right|_{\varphi\to \theta_{\mathrm e}}\\
  &=f_{\rv \varphi}(\varphi)\left(\frac{\frac{d}{t_{\mathrm c}v}\cos[\theta_{\mathrm e}]}{\sqrt{1-\frac{d^2}{(t_{\mathrm c}v)^2}\sin^2[\theta_{\mathrm e}]}}+1\right)_{\varphi\to \theta_{\mathrm e}}\\&
  =\frac{1}{2\pi}\left(\frac{\frac{d}{t_{\mathrm c}v}\cos[\theta_{\mathrm e}]}{\sqrt{1-\frac{d^2}{(t_{\mathrm c}v)^2}\sin^2[\theta_{\mathrm e}]}}+1\right).
	\end{aligned}
\end{equation}\normalsize
$\rv \varphi$ is in the range $[0,2\pi]$ implying that $-\frac{d}{t_{\mathrm c}v}\sin(\rv \theta_{\mathrm e})$ is in the range of $[-1,1]$ and then  $$\small\begin{aligned} &\rv \theta_{\mathrm e}\in\left[-\arcsin\left(\frac{t_{\mathrm c}v}{d}\right),\arcsin\left(\frac{t_{\mathrm c}v}{d}\right)\right] \\&\qquad\qquad\cup \left[\pi-\arcsin\left(\frac{t_{\mathrm c}v}{d}\right),\pi+\arcsin\left(\frac{t_{\mathrm c}v}{d}\right)\right]\end{aligned}.$$\normalsize
The MD displacement distance is a fraction of the BS-MD distance, which gives $\rv \theta_{\mathrm e}\in[-\pi/4,\pi/4]$ and $\arcsin(t_{\mathrm c}v/d)\in [-\pi/4,\pi/4]$. These imply $$\rv \theta_{\mathrm e}\notin \left[\pi-\arcsin\left(\frac{t_{\mathrm c}v}{d}\right),\pi+\arcsin\left(\frac{t_{\mathrm c}v}{d}\right)\right].$$ Knowing that $\rv \theta_{\mathrm e}$ is only in the range of the left side of the union, the probability of $\rv \theta_{\mathrm e}$ becomes
\begin{equation}
	\begin{aligned}
		f_{\rv \theta_{\mathrm e}}(\theta_{\mathrm e})
		&=\frac{1}{\pi}\left(\frac{\frac{d}{t_{\mathrm c}v}\cos[\theta_{\mathrm e}]}{\sqrt{1-\frac{d^2}{(t_{\mathrm c}v)^2}\sin^2[\theta_{\mathrm e}]}}+1\right).
	\end{aligned}
\end{equation}
Owing the expression of $f_{\rv \theta_{\mathrm e}}(\theta_{\mathrm e})$, we can now derive an expression of the $J_{\mathrm{DCT}}(t_{\mathrm c})$ as follows.
\begin{equation}\small
	\begin{aligned}
	&J_{\text{DCT}}(t_{\mathrm c})\\&=\frac{1}{B}\left|\int_{0}^{2\pi}\int_{-\frac{ B}{2}}^{\frac{ B}{2}}f_{\rv \varphi}(\varphi)\exp[\jmath 2\pi f\left(\tau_2-\tau_{1}-(\rv\tau'_2-\rv\tau'_{1})\right)]\,{d}f{d}\varphi\right|
	\\&=\frac{1}{B}\Bigg|\int_{-\arcsin\left(\frac{t_{\mathrm c}v}{d}\right)}^{\arcsin\left(\frac{t_{\mathrm c}v}{d}\right)}\int_{-\frac{ B}{2}}^{\frac{ B}{2}}f_{\rv \theta_{\mathrm e}}(\theta_{\mathrm e})\\&\qquad\qquad \exp\Big[\jmath 2\pi \frac{f}{c} \Big(q_1-q_1\cos[\theta_{\mathrm e}]-q_2\sin[\theta_{\mathrm e}]\Big)\Big]{d}f{d}\theta_{\mathrm e}\Bigg|
	\\&=\frac{1}{B\pi}\Bigg|\int_{-\arcsin\left(\frac{t_{\mathrm c}v}{d}\right)}^{\arcsin\left(\frac{t_{\mathrm c}v}{d}\right)}\int_{-\frac{ B}{2}}^{\frac{ B}{2}}\left(\frac{\frac{d}{t_{\mathrm c}v}\cos[\theta_{\mathrm e}]}{\sqrt{1-\frac{d^2}{(t_{\mathrm c}v)^2}\sin^2[\theta_{\mathrm e}]}}+1\right)\\
 &\qquad\qquad\exp\Big[\jmath 2\pi \frac{f}{c} \Big(q_1-q_1\cos[\theta_{\mathrm e}]-q_2\sin[\theta_{\mathrm e}]\Big)\Big]{d}f{d}\theta_{\mathrm e}\Bigg|
		\\&=\frac{B}{B\pi}\Bigg|\int_{-\arcsin\left(\frac{t_{\mathrm c}v}{d}\right)}^{\arcsin\left(\frac{t_{\mathrm c}v}{d}\right)}\left(\frac{\frac{d}{t_{\mathrm c}v}\cos[\theta_{\mathrm e}]}{\sqrt{1-\frac{d^2}{(t_{\mathrm c}v)^2}\sin^2[\theta_{\mathrm e}]}}+1\right)\\
  &\qquad\qquad\sinc\Big[\pi \frac{ B}{c} \Big(q_1-q_1\cos[\theta_{\mathrm e}]-q_2\sin[\theta_{\mathrm e}]\Big)\Big]{d}\theta_{\mathrm e}\Bigg|.
  \end{aligned}
\end{equation}\normalsize
Let us consider the random variable $\rv z=\sin(\rv \theta_{\mathrm e})$, for $\rv \theta_{\mathrm e}\in \left[-\arcsin\left(\frac{t_{\mathrm c}v}{d}\right),\arcsin\left(\frac{t_{\mathrm c}v}{d}\right)\right]\subset[-\pi/4,\pi/4].$ In this angle range, we have $\cos[\rv \theta_{\mathrm e}]=\sqrt{1-\sin^2[\rv \theta_{\mathrm e}]}=\sqrt{1-\rv z^2}.$
Substituting the variable $\sin[\rv \theta_{\mathrm e}]$ for $z$ in the expression of the coherence time function gives 
\begin{equation}\small
	\begin{aligned}
J_{\text{DCT}}(t_{\mathrm c})&=\frac{1}{\pi}\int_{-\frac{t_{\mathrm c}v}{d}}^{\frac{t_{\mathrm c}v}{d}}\left(\frac{\frac{d}{t_{\mathrm c}v}\sqrt{1-z^2}}{\sqrt{1-\frac{d^2}{(t_{\mathrm c}v)^2}z^2}}+1\right)\\&
\qquad\sinc\Big[\pi \frac{ B}{c} \Big(q_1-q_1\sqrt{1-z^2}-q_2z\Big)\Big]\frac{1}{\frac{\partial \sin[\theta_{\mathrm e}]}{\partial \theta_{\mathrm e}}}\,{d}z
\\&=\frac{1}{\pi}\int_{-\frac{t_{\mathrm c}v}{d}}^{\frac{t_{\mathrm c}v}{d}}\left(\frac{\frac{d}{t_{\mathrm c}v}\sqrt{1-z^2}}{\sqrt{1-\frac{d^2}{(t_{\mathrm c}v)^2}z^2}}+1\right)\\
&\qquad
\sinc\Big[\pi \frac{ B}{c} \Big(q_1-q_1\sqrt{1-z^2}-q_2z\Big)\Big]\frac{1}{\cos[\theta_{\mathrm e}]}\,{d}z
\\&=\frac{1}{\pi}\int_{-\frac{t_{\mathrm c}v}{d}}^{\frac{t_{\mathrm c}v}{d}}\left(\frac{\frac{d}{t_{\mathrm c}v}\sqrt{1-z^2}}{\sqrt{1-\frac{d^2}{(t_{\mathrm c}v)^2}z^2}}+1\right)\\
&\qquad\sinc\Big[\pi \frac{ B}{c} \Big(q_1-q_1\sqrt{1-z^2}-q_2z\Big)\Big]\frac{1}{\sqrt{1-z^2}}\,{d}z
\\&=\frac{1}{\pi}\int_{-g_{\max}}^{g_{\max}}\left(\frac{1}{g_{\max}\sqrt{1-\frac{1}{g^2_{\max}}z^2}}+\frac{1}{\sqrt{1-z^2}}\right)\\
&\qquad\sinc\Big[\pi \frac{ B}{c} \Big(q_1-q_1\sqrt{1-z^2}-q_2z\Big)\Big]\,{d}z,
 \end{aligned}
\end{equation}
where $g_{\max}=\frac{t_{\mathrm c}v}{d}$.
This ends the proof of Theorem \ref{Lem:PositionCoherenceTime}.

\section{Proof Of Lemma \ref{Lem:LowerBoundJDCT}}\label{Proof:lemJDCTLowerBound}

For $l_1=l_2=l$ we have \begin{equation}\vspace{-0.1cm}
	\begin{cases}
		-2l\leq q_2,q_1\leq 2l\\
		|1-z-\sqrt{1-z^2}|\leq 2|z|\,\, for |z|\leq 1
	\end{cases}. \end{equation}

Recall that $|z|\leq g_{\max}<1$ and hence
$$\pi \frac{ B}{c} |q_1(1-\sqrt{1-z^2})-q_2z|\leq 4\pi l \frac{ B}{c}|z|.$$
Accordingly, we get:
\begin{subequations}\small\label{eq:LowerBoundPositionCoherenceFunction}
	\begin{align}
		&J_{\mathrm{DCT}}(t_{\mathrm c})\nonumber
  \\&\geq\frac{1}{\pi}\left|\int_{-g_{\max}}^{g_{\max}}\left(\frac{1}{g_{\max}\sqrt{1-\frac{z^2}{g^2_{\max}}}}+\frac{1}{\sqrt{1-z^2}}\right)\sinc[4\pi l\frac{B}{c}z]{d}z\right|\\
		&=\frac{1}{\pi}\Bigg|\int_{-\frac{lB}{c}g_{\max}}^{\frac{lB}{c}g_{\max}}\Bigg(\frac{1}{\frac{lB}{c}g_{\max}\sqrt{1-\frac{1}{\left(\frac{lB}{c}g_{\max}\right)^2}z^2}}\nonumber\\
&\qquad\qquad\qquad+\frac{1}{\frac{lB}{c}\sqrt{1-\left(\frac{c}{lB}z\right)^2}}\Bigg)\sinc[4\pi z]\,{d}z\Bigg|\label{eq:LBoundDCTSecondLine}\\
		&\simeq \frac{1}{\pi}\Bigg|\int_{-\frac{lB}{c}g_{\max}}^{\frac{lB}{c}g_{\max}}\Bigg(\frac{1}{\frac{lB}{c}g_{\max}\sqrt{1-\frac{1}{\left(\frac{lB}{c}g_{\max}\right)^2}z^2}}\nonumber\\
&\qquad\qquad\qquad+\frac{1}{\frac{lB}{c}\sqrt{1-\left(\frac{c}{lB}z\right)^2}}\Bigg) \cos[2\pi z]\,{d}z\Bigg|\\
		&\geq \frac{1}{\pi}\Bigg|\int_{-\frac{lB}{c}g_{\max}}^{\frac{lB}{c}g_{\max}}\Big(\frac{1}{\frac{lB}{c}g_{\max}\sqrt{1-\frac{1}{\left(\frac{lB}{c}g_{\max}\right)^2}z^2}}\Big)\cos[2\pi z]\,{d}z\Bigg|\\&
		=\left|J_{0}\left(2\pi \frac{lB}{c} g_{\max}\right)\right|,
	\end{align}
\end{subequations}\normalsize
where (\ref{eq:LBoundDCTSecondLine}) is the result of change of variables, in which we substitute $z$ by $\frac{lB}{c}z$. 

\bibliographystyle{IEEEtran}
\bibliography{IEEEabrv,WGroup}

\begin{thebibliography}{10}
\providecommand{\url}[1]{#1}
\csname url@samestyle\endcsname
\providecommand{\newblock}{\relax}
\providecommand{\bibinfo}[2]{#2}
\providecommand{\BIBentrySTDinterwordspacing}{\spaceskip=0pt\relax}
\providecommand{\BIBentryALTinterwordstretchfactor}{4}
\providecommand{\BIBentryALTinterwordspacing}{\spaceskip=\fontdimen2\font plus
\BIBentryALTinterwordstretchfactor\fontdimen3\font minus \fontdimen4\font\relax}
\providecommand{\BIBforeignlanguage}[2]{{%
\expandafter\ifx\csname l@#1\endcsname\relax
\typeout{** WARNING: IEEEtran.bst: No hyphenation pattern has been}%
\typeout{** loaded for the language `#1'. Using the pattern for}%
\typeout{** the default language instead.}%
\else
\language=\csname l@#1\endcsname
\fi
#2}}
\providecommand{\BIBdecl}{\relax}
\BIBdecl

\bibitem{HexaX}
\BIBentryALTinterwordspacing
M.~{E. Leinonen et al.}, ``Initial radio models and analysis towards ultra-high data rate links in {6G},'' \emph{Hexa-X D2.2 v1.0.}, Dec. 2021. [Online]. Available: \url{https://hexa-x.eu/wp-content/uploads/2022/01/Hexa-X-D2_2.pdf}
\BIBentrySTDinterwordspacing

\bibitem{HexaX1}
M.~{A. Uusitalo et al.}, ``{6G} vision, value, use cases and technologies from european {6G} flagship project {Hexa-X},'' \emph{{IEEE} Access}, Nov. 2021.

\bibitem{PhaseNoise3GPP}
\BIBentryALTinterwordspacing
3GPP, ``Technical specification group radio access network; study on supporting {NR from 52.6 GHz to 71 GHz} ({Release 17}),'' 3GPP, TR 38.808, March 2021. [Online]. Available: \url{https://www.3gpp.org/ftp/Specs/archive/38_series/38.808/38808-h00.zip}
\BIBentrySTDinterwordspacing

\bibitem{DuVal}
J.~{Du} and R.~A. {Valenzuela}, ``How much spectrum is too much in millimeter wave wireless access,'' \emph{{IEEE} J. Sel. Areas Commun.}, vol.~35, no.~7, pp. 1444--1458, 2017.

\bibitem{ChiVenVal:3}
D.~Chizhik, S.~Venkatesan, and R.~A. Valenzuela, ``{Physical limits on beam switching performance of LOS mmWave links},'' \emph{Bell Labs, Murray Hill, NJ, USA, Tech. Rep.}, vol. ITD-15-55823C, 2015.

\bibitem{ChiVenVal:1}
D.~Chizhik, S.~Venkatesan, R.~A. Valenzuela, and S.~Wilkus, ``{Assessment of viability and challenges facing mmWave communications},'' \emph{Bell Labs, Murray Hill, NJ, USA, Tech. Rep.}, vol. ITD-13-54539G, 2013.

\bibitem{VaChoiHeath}
V.~{Va}, J.~{Choi}, and R.~W. {Heath}, ``The impact of beamwidth on temporal channel variation in vehicular channels and its implications,'' \emph{{IEEE} Trans. Veh. Technol.}, vol.~66, no.~6, pp. 5014--5029, 2017.

\bibitem{TelTse:00}
I.~E. Telatar and D.~Tse, ``Capacity and mutual information of wideband multipath fading channels,'' \emph{{IEEE} Trans. Inf. Theory}, Jul. 2000.

\bibitem{ZhaZhaZha}
W.~{Zhang}, W.~{Zhang}, and S.~{Zhang}, ``Location information based beam training for {UAV} {mmWave} system,'' in \emph{Proc IEEE ICCC}, 2019.

\bibitem{MohanedOzgur}
M.~Chraiti and O.~Ercetin, ``On the space/time correlation of mmwave aoas: Concept and experimental validation,'' in \emph{Proc. IEEE WCNC}, 2024.

\bibitem{PhaseNoise:ChungPra}
M.~Chung, H.~Prabhu, F.~Sheikh, O.~Edfors, and L.~Liu, ``Low-complexity fully-digital phase noise suppression for millimeter-wave systems,'' in \emph{Proc. IEEE ISCAS}, 2020.

\bibitem{PhaseNoise:Song}
H.-G. Song, K.~Park, J.-Y. Park, T.-H. Kwon, J.-S. Seo, and S.-W. Jeon, ``{5G NR} performance evaluation under phase noise distortion for {52.6 GHz to 71 GHz},'' in \emph{Proc. ICTC}, 2021.

\bibitem{MohanedDidem}
D.~Aydoğan, M.~Chraiti, and K.~K. Tokgöz, ``A comparative study on phase noise model in ultra-high data-rate sub-thz communications,'' in \emph{Proc. IEEE EuCNC/6G Summit}, 2025.

\bibitem{MohanedNFreqMulti}
A.~Alperen~Oznam, C.~Farhati, M.~Chraiti, A.~Ghrayeb, H.~Celebi, F.~Abdelkefi, and K.~Kaan~Tokgoz, ``Bandwidth expansion in n-fold frequency multiplier: Is it n or $\sqrt{N}$?'' \emph{{IEEE} Commun. Lett.}, vol.~29, no.~6, pp. 1205--1209, 2025.

\bibitem{RenHaaGhr}
M.~{Di Renzo}, H.~{Haas}, A.~{Ghrayeb}, S.~{Sugiura}, and L.~{Hanzo}, ``Spatial modulation for generalized {MIMO}: Challenges, {O}pportunities, and {I}mplementation,'' \emph{Proc. {IEEE}}, vol. 102, no.~1, pp. 56--103, 2014.

\bibitem{JegGhrSzc:J1}
J.~{Jeganathan}, A.~{Ghrayeb}, L.~{Szczecinski}, and A.~{Ceron}, ``Space shift keying modulation for {MIMO} channels,'' \emph{{IEEE} Trans. Wireless Commun.}, vol.~8, no.~7, pp. 3692--3703, 2009.

\bibitem{JegGhrSzc:J2}
J.~{Jeganathan}, A.~{Ghrayeb}, and L.~{Szczecinski}, ``Spatial modulation: optimal detection and performance analysis,'' \emph{IEEE Communications Letters}, vol.~12, no.~8, pp. 545--547, 2008.

\bibitem{JDAS1}
A.~Moerman, J.~Van~Kerrebrouck, O.~Caytan, I.~L. de~Paula, L.~Bogaert, G.~Torfs, P.~Demeester, H.~Rogier, and S.~Lemey, ``Beyond 5g without obstacles: mmwave-over-fiber distributed antenna systems,'' \emph{IEEE Communications Magazine}, vol.~60, no.~1, pp. 27--33, 2022.

\bibitem{JDAS2}
M.~U. Sheikh, K.~Ruttik, and R.~Jantti, ``Das and udn solutions for indoor coverage at millimeter wave (mmwave) frequencies,'' in \emph{Proc. IEEE VTC}, 2019.

\bibitem{ORAN_nGRG_2024_mWAD}
\BIBentryALTinterwordspacing
M.~Alavirad, H.~Miyata, T.~Funada, M.~Suga, H.~Toshinaga, N.~Kita, R.~Inohara, and K.~Tanaka, ``Use case analysis on mmwave antenna distribution (mwad),'' O{-}RAN ALLIANCE, nGRG Contributed Research Report RR-2024-07, Jul. 2024, release date: 2024-07-09; version v1.0.3. [Online]. Available: \url{https://mediastorage.o-ran.org/ngrg-rr/nGRG-RR-2024-04-UC_analysis_mmWave_antenna_dist-v1_0_3.pdf}
\BIBentrySTDinterwordspacing

\bibitem{Mohaned}
M.~Chraiti, A.~Conti, and M.~Z. Win, ``{mmWave} communications for high mobility devices: The case of road side links,'' in \emph{Proc. IEEE SPAWC}, 2021.

\bibitem{QiKobSud:VT}
Y.~{Qi}, H.~{Kobayashi}, and H.~{Suda}, ``On time-of-arrival positioning in a multipath environment,'' \emph{{IEEE} Trans. Veh. Technol.}, vol.~55, no.~5, pp. 1516--1526, 2006.

\bibitem{DavTsePraVis:B2005}
T.~David and P.~Viswanath, \emph{Fundamnetals of Wireless Communication}.\hskip 1em plus 0.5em minus 0.4em\relax Cambridge, MA, USA: Cambridge University Press, 2005.

\bibitem{Rap:96}
T.~S. Rappaport, \emph{Wireless communications: principles and practice}.\hskip 1em plus 0.5em minus 0.4em\relax prentice hall PTR New Jersey, 1996, vol.~2.

\bibitem{Car:87}
G.~C. {Carter}, ``Coherence and time delay estimation,'' \emph{Proc. {IEEE}}, vol.~75, no.~2, pp. 236--255, 1987.

\bibitem{Rap:B02}
T.~Rappaport, \emph{Wireless Communications: Principles and Practice}.\hskip 1em plus 0.5em minus 0.4em\relax Upper Saddle River, NJ, USA: Prentice Hall PTR, 2002.

\bibitem{JPhaseNoiseWiener}
H.~Ghozlan and G.~Kramer, ``Models and information rates for wiener phase noise channels,'' \emph{{IEEE} Trans. Inf. Theory}, vol.~63, no.~4, pp. 2376--2393, 2017.

\end{thebibliography}

\end{document}